\documentclass{article}
\usepackage[colorlinks=true,allcolors=blue]{hyperref}
\usepackage{graphicx} 
\usepackage{caption}
\usepackage{subcaption}
\usepackage{comment}
\usepackage{float}
\usepackage{appendix}
\usepackage{amsmath,amssymb,amsthm,amsfonts}
\usepackage{fullpage,times,fourier,charter}
\usepackage{bm}
\usepackage{xcolor}
\usepackage[backend=biber,style=numeric,doi=true,isbn=false,url=true]{biblatex}
\graphicspath{{figures/}}

\allowdisplaybreaks

\makeatletter
\renewcommand\section{\@startsection {section}{1}{\z@}%
  {-2.4ex \@plus -1ex \@minus -.2ex}{1.2ex \@plus.1ex}%
  {\normalfont\bf\sffamily}}
\renewcommand\subsection{\@startsection{subsection}{2}{\z@}%
  {-2.0ex\@plus -0.4ex \@minus -.2ex}{1.0ex \@plus .1ex}%
  {\normalfont\small\bf\sffamily}}

\def\sn{\mathop{\rm sn}\nolimits}

\newtheorem{proposition}{Proposition}[section]
\newtheorem{remark}{Remark}[section]

\newcommand\partialderiv[3][]{\frac{\partial^{#1}#2}{\partial #3^{#1}}}
\def\sech{\mathop{\rm sech}\nolimits}
\def\arcsinh{\mathop{\rm arcsinh}\nolimits}

\numberwithin{equation}{section}

\def\be{\begin{equation}}
\def\ee{\end{equation}}
\def\bse{\begin{subequations}}
\def\ese{\end{subequations}}

\advance\textheight 1mm
\advance\voffset 0mm
\advance\parskip 1pt

\title{Wedge problems and dispersive shock waves in the two-dimensional Toda lattice}
\author{Marco Calabrese$^1$, Gino Biondini$^2$, Christopher Chong$^3$ and Panayotis Kevrekidis$^1$\\[1ex]
\normalsize\it 
1: Department of Mathematics and Statistics, University of Massachussetts, Amherst, MA\\
\normalsize\it 
2: Department of Mathematics, State University of New York at Buffalo, Buffalo, NY\\
\normalsize\it 
3: Department of Mathematics, Bowdoin College, Brunswick, ME
}
\date{\small\today}

\begin{document}

\maketitle
\begin{abstract}

We study the formation and interaction of dispersive shock waves (DSWs) in the two-dimensional Toda lattice subject to wedge-type initial conditions, 
and show that their interaction gives rise to a discrete analog of Mach reflection for 
dispersive shock waves in discrete systems. 
The initial jump across each  leg  of the wedge acts locally as a Riemann problem for the one-dimensional Toda lattice, producing two oblique DSWs whose leading-edge soliton amplitude is determined explicitly by the one-dimensional Whitham modulation theory. 
The two-dimensional nature of the problem manifests when these oblique DSWs meet along the symmetry axis.
We show that,  for compressive wedges (i.e., when the initial conditions are such that two oblique DSWs that are generated propagate toward each other),
a critical slope $q_{\mathrm{cr}}$ separates two regimes: in the subcritical regime ($q<q_{\mathrm{cr}}$) the interaction is resonant and it produces an expanding DSW whose amplitude, length and velocity are explicitly computed by using exact soliton solutions of the two-dimensional Toda lattice; in the supercritical regime ($q>q_{\mathrm{cr}}$) the interaction is ordinary and produces a localized peak whose amplitude is determined analytically. 
We also show qualitatively that a similar dichotomy between two regimes exists for expansive wedges (i.e., when the initial conditions are such that the two oblique DSWs propagate away from each other). 
We confirm all analytical predictions by comparing them with the results of direct  numerical simulations. 
 Finally, we show that the 
continuum limit of the result is consistent with the analogous theory for the Kadomtsev-Petviashvili equation, providing an independent validation of the analytical framework.    
\end{abstract}
\tableofcontents

\bigskip
\section{Introduction}
\label{s:intro}

The study of dispersive shock waves (DSWs) has a rich and multifaceted history. DSWs are expanding oscillatory wave trains that can emerge from the dispersive regularization of wave breaking. In the prototypical setting of a Riemann problem, namely for step-like initial conditions connecting two distinct constant states, a DSW forms and expands as a modulated periodic wave train. 
In many cases, this modulation can be effectively described quantitatively using Whitham modulation theory \cite{Whitham1974,EL201611}. 
In this framework, the DSW can be seen as a modulated genus-one solution of the underlying equation, whose slowly varying parameters are governed by a system of hyperbolic conservation laws for the branch points of the associated spectral curve.
In the prototypical KdV-type scenario relevant to the present work, the wave train degenerates into small amplitude linear waves at one edge (called the trailing edge) and into a large solitary wave at the other (called the leading edge) where modulation equations yield explicit formulae for its speed and its amplitude. Significant progress has been made in the theory of one-dimensional DSWs, both through analytical developments \cite{EL201611, MINZONI20166, MILLER201666, ABLOWITZ201684} and through numerous experimental realizations in nonlinear optics \cite{wan2007dispersive,PhysRevLett.118.254101, PhysRevLett.126.183901}, Bose-Einstein condensates \cite{PhysRevA.74.023623, PhysRevA.80.043606, PMID:30405131} and fluid dynamics \cite{Maiden_Franco_Webb_El_Hoefer_2020,  PhysRevLett.116.174501, PhysRevLett.120.144101}. By contrast, the study of DSWs in two spatial dimensions remains comparatively less well developed.

The study of two-dimensional DSWs is a more recent frontier that has only lately begun to be systematically explored \cite{PhysRevA.80.061601,BBH_PRL2025}. 
A noteworthy contribution in this respect is the recent work \cite{BBH_PRL2025}, which investigated Riemann type problems for the Kadomtsev-Petviashvili  (KP) equation, i.e., the prototypical integrable model for weakly two-dimensional, weakly nonlinear waves \cite{AblowitzSegur1981,NMPZ1984}. The authors of \cite{BBH_PRL2025} considered wedge-type initial conditions (ICs) and demonstrated the existence of a critical wedge angle separating two qualitatively distinct dynamical regimes: a subcritical regime in which two oblique DSWs interact resonantly to produce an expanding vertical stem DSW, and a supercritical regime in which the interaction is ordinary and no stem is formed. These results reveal a rich and novel phenomenology that has no counterpart in one spatial dimension, and raise the natural question of the potential
universality of such features. In particular, a notable question 
concerns whether analogous phenomena arise in other two-dimensional dispersive systems, and more specifically in discrete ones.

A parallel active research direction concerns DSWs in spatially discrete systems, where the dispersion arises from the lattice structure. While much of the literature on DSWs is rooted in continuous models, recent theoretical and experimental developments have demonstrated that discrete metamaterials can exhibit genuinely discrete DSW features \cite{:/publisher/EDPSciences/journal/EuropeanPhysicalJournalB/14/1/10.1007/s100510050119, PhysRevE.80.056602, nesterenko2013dynamics}; indeed, relevant 
experiments have lately emerged  both in mechanical settings such as hollow
ellpitic cylinders~\cite{kimyang}, as well as in magnetic ones
in~\cite{talcohen}. Among discrete models, integrable systems occupy a privileged role: the availability of exact analytical tools (Lax pairs, inverse scattering transform, multi-soliton solutions and Whitham modulation equations) enables a precise  analytical characterization of the wave dynamics that can serve as a paradigm for understanding more general, non-integrable discrete systems. The one-dimensional Toda lattice (1DTL) is the prototypical example in this class, and its Whitham modulation theory and DSWs were studied in full detail \cite{AMBYK,GBCCPK}.
Very recently, quasi-one-dimensional DSWs were studied in a two-dimensional analog of the Fermi–Pasta–Ulam chain \cite{chong2026travelingdispersiveshockwaves}. That work focused exclusively on line initial data. The study of genuinely two-dimensional dynamics, such as those arising from a wedge problem, is a natural next step.

The present work addresses precisely this focal point of
discrete dispersive shock waves in two dimensions by bringing together three active research directions: two-dimensional dispersive shock waves, discrete nonlinear systems, and integrability. 
We consider a prototypical semi-discrete integrable model in (2+1) dimensions: the two-dimensional Toda lattice (2DTL), which has been studied extensively in the literature (e.g., see \cite{BW2010,HirotaItoKako,HirotaOhtaSatsuma,HirotaOhtaSatsuma_PTPS,KajiwaraSatsuma,JPhysA2004,Mikhailov1979,Mizumachi2025,Zhang2002} and references therein).

The 2DTL appears in several equivalent formulations  in the literature. One common representation of the model is given in the light-cone coordinates $(r,s)$  (see, e.g.,  \cite{hirota1988two, BW2010})
\vspace*{-0.4ex}
\begin{equation}\label{e:Toda2D_0}
\frac{\partial^2}{\partial r \partial s}\log (1+V_n)=V_{n+1}-2V_n+V_{n-1},  
\end{equation}
where $n\in \mathbb{N}$ is the discrete lattice index. For the purposes of the present work, however, it will be more convenient to  work in the laboratory coordinates $(y,t)$, where $y\in \mathbb{R}$ denotes a continuous transverse direction and $t$ denotes time. 
We therefore perform the invertible change of independent variables 
\begin{equation}
\label{e:labframecoordinates}
  r = \tfrac12{(t+y)},\qquad 
  s = \tfrac12{(t-y)}.
\end{equation}
Using~\eqref{e:labframecoordinates}, Eq. \eqref{e:Toda2D_0} becomes
\begin{equation} \label{2dTodaLog}
 \left ( \frac{\partial^2}{\partial t^2}- \frac{\partial^2}{\partial y^2} \right )\log{(1+V_n)}=V_{n+1}-2V_n+V_{n-1}\,
\end{equation}
which is the form of the 2DTL directly connected to the $y$ independent reduction \eqref{1dTodalog} discussed below in Section \ref{s: PeriodicSolutions2d}.
Introducing the strain variable 
$B_n=\log{(1+V_n)}$,  Eq.  \eqref{2dTodaLog}  becomes:  
\vspace*{-0.4ex}
\begin{equation} \label{eq_Toda1}
\partialderiv[2]{B_n}t - \partialderiv[2]{B_n}y = e^{B_{n+1}}-2e^{B_n}+e^{B_{n-1}} \,.
\end{equation}
We adopt this formulation for two reasons. First, it naturally reduces to the one-dimensional Toda lattice along oblique traveling-wave reductions in the variable $k_1n+k_2 y+\Omega t+\Phi_0$. 
Second, in the long-wave, small-amplitude limit, the wedge problem considered in this work reduces to the corresponding wedge problem for the KP equation, providing a direct connection between the discrete and continuum frameworks and between our work and the recent
one of~\cite{BBH_PRL2025}. 

\begin{figure}[b!]
\centerline{\includegraphics[width=0.7\textwidth]{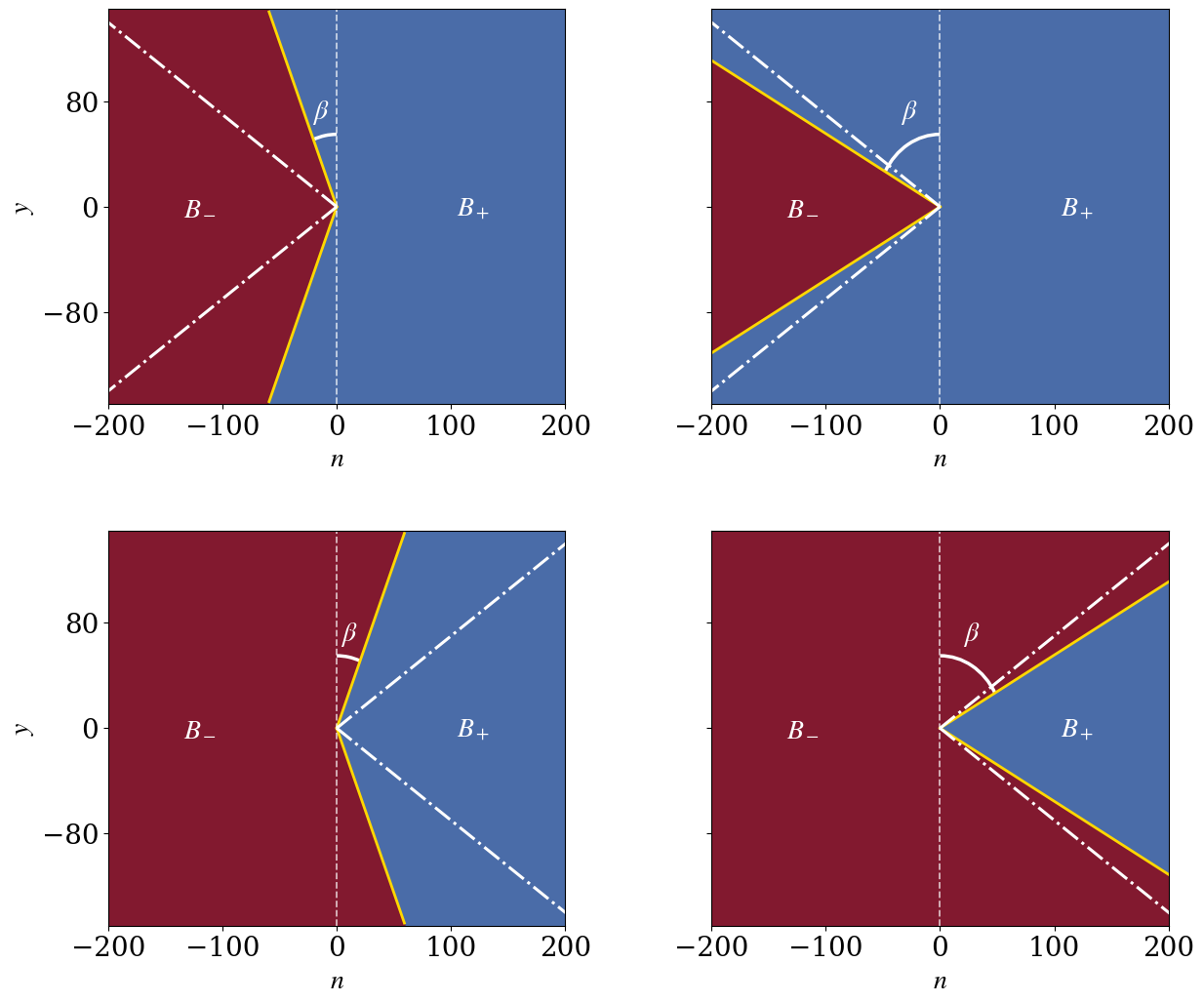}}
\caption{The four dynamical regimes of the wedge problem~\eqref{InitialValue2dToda}, with $B_-=0.5$  (dark red) and $B_+=0$ (blue) as in the numerical simulations of Section  \ref{s:wedge}. Top row: type I wedges, with $q=0.4$ (subcritical, left) and $q=1.8$ (supercritical, right). Bottom  row: type II wedges, with $q=-0.4$ (subcritical, left) and $q=-1.8$ (supercritical, right). The gold line marks the wedge boundary at the slope $q$ of each panel; the angle $\beta$ formed with the $y-$axis which satisfies $\beta=\arctan{q}$ is indicated in each panel. The vertical dashed line marks $n=0$. 
The white dash-dotted lines correspond to the critical slope $|q|= q_{\mathrm{cr}}\approx1.429$ from Eq.~\eqref{Eq:Res_qA} separating the subcritical and supercritical regimes. 
} 
\label{fig:wedge}
\vspace*{-1.4ex}
\end{figure}

Specifically, motivated by the results of~\cite{BBH_PRL2025},
in this work we study Eq. \eqref{eq_Toda1} subject to wedge-type ICs of the form
\begin{equation} \label{InitialValue2dToda}
    B_n(y,0)=\begin{cases}
        B_-, &n + q\,|y| \le 0,\\
        B_+, &n + q\,|y| > 0,
    \end{cases}
\end{equation}
where $B_->B_+$ are constants and $q\in \mathbb{R}$ is the slope parameter of the wedge. Since Eq. \eqref{eq_Toda1} is second-order in time, a second initial condition $\dot B_n(y,0)$ must also be prescribed; this is chosen so that waves are right-moving, see Eq.  \eqref{InitialVelocity} below. 
Depending  on the sign of~$q$, the initial data take two geometrically distinct forms, as illustrated in Figure~\ref{fig:wedge}. 
The angle $\beta$ formed by the wedge with the positive $y$ axis, counted counterclockwise,
is determined by the slope parameter as $\beta = \tan q$.
Hereafter, similarly to~\cite{BBH_PRL2025}, we will denote as type~I wedges those ICs with $q>0$
(as in the right panel of Fig.~\ref{fig:wedge})
and as type~II wedges those with $q<0$
(as in the left panel of Fig.~\ref{fig:wedge}).

For type~I wedges, the region of greater amplitude $B_-$ forms a rightward-pointing wedge in the $ny$-plane, whereas for type~II wedges it forms a leftward-pointing wedge.
Away from the line $y=0$,
the jump across each branch of the wedge acts locally as a Riemann problem for the one-dimensional Toda lattice. 
We will show that, as a consequence, the dynamics initially generates two oblique DSWs propagating symmetrically away from the wedge boundaries. The leading-edge soliton amplitude of each oblique DSW is completely determined by the one-dimensional Whitham modulation theory.
Since the dynamics is purely right-moving, the two oblique DSWs generated at the wedge boundaries both propagate towards increasing $n$; depending on the wedge orientation, this causes them to converge along the symmetry axis $y=0$ for type II wedges, or to diverge from one another for type I wedges. This simple observation anticipates why type II wedges give rise to a genuinely two-dimensional interaction, the central object of study of this work, whereas type I wedges exhibit a comparatively simpler phenomenology, discussed only briefly in Section \ref{TypeI}. The truly two-dimensional nature of the problem manifests itself when these oblique DSWs interact along the symmetry axis $y=0$. The outcome of this interaction
depends critically on the slope parameter $q$ and is the central object of study of this work.

The main results of this work can be summarized as follows. For type II wedges, we show that a critical value $q_{\mathrm{cr}}$  separates two dynamically distinct regimes. In the subcritical regime $|q|<q_{\mathrm{cr}}$
the two oblique DSWs interact resonantly and produce an expanding vertical stem DSW:  leveraging the exact  two-soliton  solutions of the 2DTL \cite{BW2010}, we derive explicit analytical expressions for the stem amplitude, its vertical expansion rate and its horizontal velocity,  and confirm all  predictions against numerical simulations. In the supercritical regime $|q|>q_{\mathrm{cr}}$, the interaction is ordinary, i.e., the two oblique solitons simply cross without generating additional structure, and produce instead a localized peak of high amplitude propagating along $y=0$, whose amplitude we also determine analytically.
For type I wedges, the interaction is less intricate, but the subcritical and supercritical regimes remain qualitatively distinct, separated by the same critical value $q_{\mathrm{cr}}$, up to sign, as for type II wedges, as we discuss. The critical slope $q_\mathrm{cr}$ itself is derived analytically, and the continuum limit of all results is shown to be consistent with the analogous theory for the KP equation developed in \cite{BBH_PRL2025}, providing an independent validation of the analytical framework.

The remainder of this paper is organized as follows. 
In Section \ref{s:1DTL} we recall some relevant background on the one-dimensional Toda lattice, including Flaschka variables and Whitham modulation theory, and compute the amplitude of the leading-edge soliton of the DSWs produced by a Riemann problem. 
In Section~\ref{s:2DTL_background} we present the soliton solutions of the 2DTL and derive a new class of periodic traveling wave solutions that generalizes the genus-one solutions of the 1DTL to the two-dimensional setting. 
These solutions provide the elemental building blocks for the more complicated, genuinely two-dimensional solutions studied in the remainder of the work.
Section~\ref{s:wedge} contains the main results on the wedge problem for the 2DTL: 
After establishing the local one-dimensional reduction and deriving the critical slope, we analyze the four dynamical regimes, providing analytical predictions and numerical confirmation for each. 
In Section~\ref{s:continuum} we show how the KP equation is obtained as the long-wave small-amplitude limit of the 2DTL, and we verify the consistency of the critical slope formula with the analogous result of \cite{BBH_PRL2025}. 
Finally, Section~\ref{s:conclusions} contains a brief summary and some concluding remarks.
The appendices contain a description of the numerical scheme (Appendix \ref{app:D}), additional background on the one-dimensional Toda lattice (Appendix \ref{app:A}), detailed derivations of the soliton solutions (Appendices \ref{app:B} and \ref{app:C}),  and a discussion of the quantitative differences between the 2DTL and the KP equation together with the extension to asymmetric wedge configurations (Appendix \ref{app:E}).

\section{Background: The one-dimensional Toda lattice} 
\label{s:1DTL}

In this section we recall relevant background material about the 1DTL that will be used in later sections to study the 2DTL, and we obtain the amplitude of the leading-edge solitons of the DSWs produced by Riemann problems.

\subsection{Flaschka variables, periodic traveling wave solutions and Whitham modulation equations}

The one-dimensional Toda lattice \cite{toda2012theory} describes an infinite chain of particles with exponential nearest neighbor interaction, governed by the equations of motion:
\begin{equation} \label{Toda1dNewton}
    \ddot{u}_n=e^{u_{n-1}-u_n}-e^{u_n-u_{n+1}},
\end{equation}
where $u_n = u_n(t)$ denotes the displacement of the $n-$th particle from its equilibrium position, dots denote differentiation with respect to time and the mass is fixed to one. 
For our purposes, an equivalent and more convenient formulation can be obtained by introducing the strain variable $B_n:=u_{n}-u_{n+1}$ in terms of which Eq. \eqref{Toda1dNewton} becomes:
\begin{equation}\label{1dTodaStainvariabel}
    \ddot{B}_n = e^{B_{n+1}} - 2 e^{B_n} + e^{B_{n-1}}.
\end{equation} 
When the solutions of~\eqref{eq_Toda1} are independent of $y$, the problem reduces to
the one-dimensional system~\eqref{1dTodaStainvariabel}. In Section~\ref{s:wedge} we show that \eqref{eq_Toda1} also reduces to \eqref{Toda1dNewton} along any oblique direction.

When a nonlinear dispersive system is given initial data with a sharp jump, the competition between nonlinearity and dispersion prevents the formation of a classical shock and instead gives rise to an expanding oscillatory wave train connecting two distinct constant states: this is a dispersive shock wave (DSW). As time progresses, the DSW expands so that its amplitude varies continuously from small-amplitude linear waves at one edge (the trailing edge) to a large solitary wave at the other (leading edge). The quantitative description of DSWs is provided by Whitham modulation theory, which replaces the original equation with a system of hyperbolic conservation laws governing the slowly varying parameters of a modulated periodic wave train. To apply this framework to the Toda model, it is convenient to reformulate the dynamics in terms of the Flaschka variables $(a_n,b_n)$, defined as in \cite{AMBYK} by
\begin{equation} \label{FlaschkaVariables}
a_n=-\dot{u}_n,\qquad 
b_n=e^{u_{n-1}-u_{n}},
\end{equation}
in terms of which the equations of motion \eqref{Toda1dNewton} become the first-order system:
\begin{gather} 
\label{Toda1dFlaschka}
    \dot{a}_n = b_{n+1}-b_n,\qquad
    \dot{b}_n = b_n(a_n-a_{n-1})\,.
\end{gather}
The strain variable is recovered through $B_n=\log{b_{n+1}}.$

Recall that the Toda lattice is a completely integrable system. 
While Whitham modulation equations can be derived for any system admitting periodic traveling waves, integrability guarantees that they take the particularly elegant structure of a diagonal system of conservation laws in Riemann invariant form, with characteristic speeds that can be computed explicitly.
Specifically, in the Flaschka coordinates, equations \eqref{FlaschkaVariables} are equivalent to the Lax representation of Toda lattice:
\begin{equation} \label{LaxEq} 
  \dot{L}+[L,E] = 0\,,
\end{equation} 
with
\begin{equation}
L = e^\partial+a_n(t)+b_n(t)e^{-\partial},\qquad 
E=e^\partial+a_n(t),
\end{equation} 
where $[X,Y]:=XY-YX$ is the Lie bracket on $\mathfrak{a}_\infty:=\{(a_{ij})_{i,j\in \mathbb{Z}} ~|~ a_{ij}=0 ~\forall~|i-j| \gg 0 \}$ and $e^\partial \psi_n:=\psi_{n+1}$ is the unit shift operator 
(for further details, please see \cite{PhysRevB.9.1924, PhysRevB.9.1921, 1974ZhETF..67..543M}). 
In turn, complete integrability guarantees the existence of an infinite family of conserved quantities and of finite-gap solutions, which in the case of~\eqref{Toda1dFlaschka} are associated with hyperelliptic curves.
Specifically, it was shown in \cite{Krichever_Novikov_1981} that finite-gap genus-$g$ solutions are naturally associated with a hyperelliptic curve $w^2=\prod_{i=1}^{2g+2}(\eta-\lambda_i)$ where $\eta$ is the spectral parameter, $w$ the associated algebraic function defining the curve, and $\lambda_i$ its branch points which coincide with the Riemann invariants. 

In the special case $g=1$, one obtains the genus-one periodic traveling wave solutions of the 1DTL, which, together with their solitonic limit, provide the building blocks for the one-dimensional and two-dimensional DSWs studied later.
The periodic solutions of the Toda lattice, derived in \cite{Teschl}, form a four-parameter family, naturally parametrized by the four brach points $\lambda_1\le\lambda_2\le\lambda_3\le\lambda_4$ of the underlying hyperelliptic curve, which coincide with the Riemann invariants of the Whitham modulation system \eqref{WhitamTodaGenusg}. 

In terms of the Flaschka variables \eqref{FlaschkaVariables}, these solutions are given by
\vspace*{-1ex}
\begin{subequations}\label{PeriodicSolutions}
\begin{align} 
    &b_{n+1}(t)=\tfrac{1}{8}\sum_{i=1}^4 \lambda_i^2+2\hat{R}_n(t)-\mu^2_n(t)-\tfrac{1}{4} a_n^2(t),\\
    &a_n(t)=\tfrac{1}{2}\sum_{i=1}^4 \lambda_i-2\mu_n(t),
    \label{e:a_n(t)_periodic}
\end{align}
\end{subequations}
where  the various quantities appearing in~\eqref{PeriodicSolutions} are as follows:
\bse
\begin{align}
    &\mu_n(t) = \frac{\lambda_2}{2} \frac{1-(\lambda_1/\lambda_2)W\,\sn^2(Z^{\mathrm{1d}}_n(t),m)}{1-W\, \sn^2(Z^{\mathrm{1d}}_n(t),m)},\qquad
    Z^{\mathrm{1d}}_n(t) = 2nF(\Xi,m)+\omega_{1d} t+Z_0,\qquad 
    \hat{R}_n(t)=-\sigma_n(t)\sqrt{P(\mu_n(t))},\qquad
    \\ 
    & P(z) =\Pi_{i=1}^4\big(z-\tfrac12 \lambda^i \big),\qquad
    m = \frac{(\lambda_3-\lambda_2)(\lambda_4-\lambda_1)}{(\lambda_4-\lambda_2)(\lambda_3-\lambda_1)}, \qquad
    \omega_{\mathrm{1d}} = \tfrac{1}{2}\sqrt{(\lambda_4-\lambda_2)(\lambda_3-\lambda_1)},\\
    &\Xi = \sqrt{\frac{\lambda_4-\lambda_2}{\lambda_4-\lambda_1}},\qquad
    W = \frac{\lambda_3-\lambda_2}{\lambda_3-\lambda_1}\,.
\end{align}
\ese

Here, $\text{sn}(z,m)$ is the Jacobi elliptic sine with parameter $0<m<1$, $F(z,m)$ is the inverse of $\text{sn}(z,m)$, $Z_0$ is a constant translation parameter, $\mu_n(t)$ is the Dirichlet eigenvalue of the scattering problem  for the Toda lattice, and $\sigma_n(t)=\pm 1$ is the sign associated with $\mu_n(t)$.

It was shown in \cite{AMBYK} that the $2g+2$ branch points $\lambda_1,\dots,\lambda_{2g+2}$ of the hyperelliptic curve above play the role of Riemann invariants for the corresponding Whitham modulation equations: 
\begin{equation}\label{WhitamTodaGenusg}
(\lambda_i)_t-s(\lambda_i) (\lambda_i)_x=0,\qquad  i=1,...,2g+2\,,
\end{equation}
where $s(\lambda_i)$ are the characteristic speeds, a general expression for which was given in \cite{AMBYK}. 
In particular, the characteristic speeds $s(\lambda_i)$ for the genus-one Whitham system are given by~\cite{AMBYK}
\begin{equation}
        s(\lambda)=\frac{(\lambda)^2-\frac{1}{2}\sigma_1\lambda+\gamma_1}{2(\lambda-\alpha_1)}\,,
\end{equation} 
with
\begin{equation}
       \alpha_1= \frac{I_1^1}{I_1^0},\qquad
       \gamma_1=-\frac{I_1^2+\gamma_0I_1^1}{I_1^0},\qquad
       \sigma_1=\sum_{i=1}^4 \lambda_i\, 
\end{equation}
and where the integrals $I_1^k$ are 
\begin{equation}
          I_1^k(\lambda_1,\lambda_2,\lambda_3,\lambda_4)=\int_0^\frac{\pi}{2}\frac{(\lambda_3-(\lambda_3-\lambda_2)\cos^2{\theta})^k}{\sqrt{((\lambda_2-\lambda_1)+(\lambda_3-\lambda_2)\sin^2{\theta})((\lambda_4-\lambda_2)-(\lambda_3-\lambda_2)\sin^2{\theta})}}d\theta\,. 
\end{equation}
Explicit expressions for the resulting speeds for the case $g=1$ were obtained in \cite{GBCCPK}.

\subsection{Riemann problems for the one-dimensional Toda lattice}
\label{s:Riemann1DTL}

In preparation for the study of Riemann problems for the two-dimensional Toda lattice that
will be carried out in later sections,
we now briefly turn our attention to the Riemann problem for \eqref{Toda1dFlaschka}, following the approach of \cite{AMBYK,GBCCPK}. 
We consider piecewise-constant initial data of the form:
\begin{equation}
\label{RiemannProblemFlaschka}
    a_n(0)=a_-\chi_{\{n\le0\}}+a_+\chi_{\{n>0\}},\qquad 
    b_n(0)=b_- \chi_{\{n\le0\}}+b_+\chi_{\{n>0\}},
\end{equation}
with $a_\pm\in \mathbb{R}$, $b_\pm >0$ and where $\chi_I$ is the characteristic function
of a set~$I\subset \mathbb{R}$, i.e.,  $\chi_I(z)=1$ if $z\in I$ and $\chi_I(z)=0$ otherwise. 
In order to express these initial data in terms of the Riemann invariants of the Whitham system, we recall that for a genus-zero solution the Riemann invariants $\lambda_1$ and $\lambda_2$ are related to the Flaschka variables \eqref{FlaschkaVariables} by
\cite{AMBYK}
\begin{equation}
    a_n(0) = \tfrac12{(\lambda_1+\lambda_2)}\, \qquad
     b_n(0) = \tfrac1{16}{(\lambda_2-\lambda_1)^2}\,.
\end{equation}
Inverting these relations, the step initial data \eqref{RiemannProblemFlaschka} correspond to a step function in the Riemann invariant space:
\begin{equation} \label{RiemannProblemLambda}
    \lambda_1=\lambda^-_1\chi_{\{n\le0\}}+\lambda^+_1\chi_{\{n>0\}},\qquad
    \lambda_2=\lambda^-_2\chi_{\{n\le0\}}+\lambda^+_2\chi_{\{n>0\}},
\end{equation}
where $\lambda^\pm_1:=a_\pm -2\sqrt{b_\pm}$ and $\lambda^\pm_2:=a_\pm+2\sqrt{b_\pm}$. 
We define the intervals $L_k:=[\min{(\lambda_k^-,\lambda_k^+)},\max{(\lambda_k^-,\lambda_k^+)}]$ for $k=1,2$, which measure the size of the jump in each Riemann invariant. It is important to note that the Toda lattice is a bidirectional system: generic step-like initial data of the form \eqref{RiemannProblemFlaschka} produce two counter-propagating DSWs, one traveling to the left and one to the right.  Depending on the relative ordering of the four values $\lambda_1^\pm,\lambda_2^\pm$, three qualitatively distinct dynamical scenarios arise:  
(i) a rarefaction wave, in which the data are non-increasing and a smooth global solution exists for all time; 
(ii) a genus-one DSW, arising when $L_1$ and $L_2$ overlap and the jump leads to shock formation regularized by a single modulated wave train; and 
(iii) a genus-two DSW, arising when $L_1$ and $L_2$ are disjoint and the regularization requires a genus-two embedding.
Next we briefly discuss the first two scenarios; 
the third one is discussed in Appendix \ref{app:A}.

\paragraph{Rarefaction waves.} 
When $\lambda_1^-\ge \lambda_1^+$ and $\lambda_2^-\ge \lambda_2^+$, the initial data \eqref{RiemannProblemLambda} are non-increasing in terms of the Whitham system \eqref{WhitamTodaGenusg}. In this case, no wave breaking occurs and a global smooth solution exists for all the time: The Riemann invariants spread out into a rarefaction wave, connecting the left and the right constant states via a smooth monotone transition. 
(For a detailed treatment of rarefaction waves in the context  of hyperbolic conservation laws, we refer the reader to \cite{Dafermos:1315649}.)

\paragraph{Genus-one regularization.} 
This scenario occurs when the intervals $L_1$ and $L_2$ overlap, i.e., when $\lambda_1^- \le \lambda_2^- \le \lambda_1^+ \le \lambda_2^+$. In this case, the initial data lead to a shock formation, which is regularized by a genus-one finite gap solution. Denoting  $(\lambda_1^-,\lambda_2^-,\lambda_1^+,\lambda_2^+) =:(\lambda_1^o,\lambda_2^o,\lambda_3^o,\lambda_4^o)$ for convenience, the solution is obtained by embedding the initial data into the degenerate genus-one Whitham system with initial condition
\begin{equation}
  \lambda_1(x,0)=\lambda_1^o,\qquad 
  \lambda_2(x,0)=\begin{cases}
        \lambda_2^o, \ \ x\le0, \\ \lambda_1^o, \ \ x>0,
    \end{cases}\qquad
  \lambda_3(x,0)= \begin{cases}
        \lambda_4^o, \ \ x\le0, \\ \lambda_3^o, \ \ x>0, 
    \end{cases}\qquad
  \lambda_4(x,0)=\lambda_4^o\,.
\end{equation}
Since these initial conditions are non-increasing, the genus-one Whitham system \eqref{WhitamTodaGenusg}  admits a global solution in $\mathbb{R}\times [0,\infty)$ given by 
\begin{equation}
    \lambda_1(x,t)=\lambda_1^o, \qquad
    \lambda_2(x,t)= \begin{cases}
            \lambda_2^o, \ \ x\le s_2^-t, \\
            \lambda_1^o,\ \ x>s_2^+t,
        \end{cases}\qquad
    \lambda_3(x,t)= \begin{cases}
            \lambda_4^o, \ \ x\le s_3^-t, \\
            \lambda_3^o,\ \ x>s_3^+t,
        \end{cases}\qquad
    \lambda_4(x,t)=\lambda_4^o,
\end{equation}
in which $\lambda_1(x,t)$ and $\lambda_4(x,t)$ remain constant for all time, while $\lambda_2$ and $\lambda_3$ vary smoothly in the expanding regions $s_2^-t < x < s_2^+t$ and $s_3^- t <x < s_3^+ t$ respectively, each connecting its two constant states. 
Outside these two regions, the solution remains equal to its initial, constant value. 
The result is two DSWs that expand between a small-amplitude linear wave and a soliton edge whose speeds are determined by $s_i$.
In particular, the speed of the right-going soliton edge is given by
\begin{equation} \label{CharSpeed}
             s_3^-=\frac{\sqrt{G(4\sqrt{b_1}+G)}}{4\log{\left( \sqrt{\frac{G}{4\sqrt{b_1}}}+\sqrt{1+\frac{G}{4\sqrt{b_1}}}\right)}},
         \end{equation}
with $G=a_2-a_1+2\sqrt{b_2}-2\sqrt{b_1}$. 
Similar expressions can be obtained for the other speeds. 
Note that the expression \eqref{CharSpeed} for $s_3^-$ is consistent with the analysis of \cite{GBCCPK} upon imposing $a_2=-a_1=2c$ and $b_1=b_2=1.$

\paragraph{Amplitude of the leading-edge solitons in the DSWs.}
We now obtain a key analytical ingredient that will be used for the characterization of two-dimensional wedge problems in Section~\ref{s:wedge}.
Namely, we relate the amplitude of the leading-edge soliton of each DSW to the initial jump in~\eqref{RiemannProblemFlaschka}.

Since the Toda lattice is bidirectional, the Riemann problem \eqref{RiemannProblemFlaschka} generically produces two counter-propagating DSWs. The amplitude of the leading-edge soliton of each DSW is obtained by combining the explicit expression \eqref{PeriodicSolutions}  with the inverse relation \eqref{RiemannProblemLambda}, and taking the appropriate solitonic limit in which one pair of brach points coalesces. 
In the rest of this section, superscripts $L$ and $R$ refer to the left- and right-going DSW, respectively, while subscripts $a$ and $b$ refer to the corresponding Flaschka variable. For the left-going DSW one finds 
\begin{align}
\label{eq: SolitonAmplitudeLeftRiemannProblem}
    A_a^L=a_++2\sqrt{b_+}-a_--2\sqrt{b_-},\qquad 
    A_b^L=  \tfrac14{(a_+-a_-+2\sqrt{b_+})^2}-b_-,
\end{align}
where we used the fact that for $n<0$ one has $\lambda_2=\lambda_2^o$ and $\lambda_3=\lambda_4^o \equiv \lambda_4$. 
Similarly, the leading edge propagating from the right has amplitudes
\begin{align}
\label{eq: SolitonAmplitudeRightRiemannProblem}
    A_a^R=a_+-2\sqrt{b_+}-a_-+2\sqrt{b_-},\qquad
    A_b^R=\tfrac14{(a_+-a_-+2\sqrt{b_-})^2}-b_+,
\end{align}
where we used the fact that for $n>0$ one has $\lambda_2=\lambda_1^o\equiv \lambda_1$ and $\lambda_3=\lambda_3^o$. 
The above expressions follow from the extremal values of the Flaschka variables,
whose derivation is given in Appendix~\ref{app:A}. 
In particular, one has
\bse
\begin{align}
    a^L_{\mathrm{max}} &=-2\sqrt{b_-}+a_++2\sqrt{b_+},\qquad
    a^L_{\mathrm{min}} =a_-,\qquad
    a^R_{\mathrm{max}} =-2\sqrt{b_+}-a_-+2\sqrt{b_-},\qquad
    a^R_{\mathrm{min}} =a_+,
\\
    b_{\mathrm{max}}^L &= \tfrac14{(a_+-a_-+2\sqrt{b_+})^2},\qquad
    b_{\mathrm{min}}^L = b_-,\qquad
    b^R_{\mathrm{max}} = \tfrac14{(a_+-a_-+2\sqrt{b_-})^2},\qquad
    b_{\mathrm{min}}^R = b_+.
\end{align}
\ese
The above amplitude formulae are the key ingredient for connecting the one-dimensional Whitham theory to the two-dimensional wedge problem that will be analyzed in Section~\ref{s:wedge}.
Indeed, as will be shown there, the leading-edge soliton of each oblique DSW generated by the two-dimensional dynamics via~\eqref{eq_Toda1} is completely characterized by these expressions, evaluated with the appropriate initial data.

\section{The two-dimensional Toda lattice and its periodic and soliton solutions} 
\label{s:2DTL_background}

In this section we present several results that will be needed for the analytical characterization of the wedge problem in Section~\ref{s:wedge}.
Specifically, we recall the soliton solutions of the 2DTL and we derive a new class of periodic traveling wave solutions that generalizes the genus-one solutions of the 1DTL of Section~\ref{s:1DTL} to the two-dimensional setting. 

\subsection{Casorati determinants, one-soliton solutions and their physical parameters} 
\label{s:OneSol}

Recall from Section~\ref{s:intro} that the 2DTL can be written in light cone coordinates $(r,s)$ and in terms of the dependent variable $V_n = V_n(r,s)$
as equation~\eqref{e:Toda2D_0}.
Also recall that, passing to the laboratory frame coordinates via the transformation 
$r=\frac12(t+y)$ and $s=\frac12(t-y)$ and introducing the strain variable $B_n=\log{(1+V_n)}$, 
the system becomes Eq. \eqref{eq_Toda1},
which is the formulation we use throughout this paper. 
The 2DTL is integrable \cite{doi:10.1142/9789814317344_0050}, a property that, as in the one-dimensional case, is the key ingredient enabling the exact analytical treatment of the wedge problem: it guarantees the existence  of multi-soliton solutions and of the tau-function theory recalled below. The 2DTL hierarchy was calculated in \cite{konopelchenko1992new,Zhang2002}.

A general framework for constructing exact solutions of 2DTL is provided by the tau function formalism
(e.g., see \cite{HirotaOhtaSatsuma_PTPS}). 
The field $V_n$ is expressed in terms of the tau functions $\tau_n$ via:
\begin{equation}
    V_n(r,s)=\frac{\partial^2} {\partial s \partial r}\log\tau_n(r,s)\,.
\end{equation}
In the laboratory frame, this relation becomes
\begin{equation} \label{Bsoliton}
    B_n(y,t)=\log \left( \frac{\tau_{n+1}\tau_{n-1}}{\tau_n^2} \right),
\end{equation}
where in the right hand side of \eqref{Bsoliton} the variables $r$ and $s$ are expressed in terms of $y$ and $t$ via \eqref{e:labframecoordinates}.
Explicit solutions are obtained by choosing $\tau_n$ as the Casorati determinant~\cite{HirotaOhtaSatsuma_PTPS}
\begin{equation}
    \tau_n^{N}(r,s)= \det \begin{pmatrix}
        f_n^{(1)} & \cdots  & f_{n+N-1}^{(1)} \\ \vdots & \ddots & \vdots  \\
         f_n^{(N)} & \cdots  & f_{n+N-1}^{(N)} 
    \end{pmatrix}
\end{equation}
where $f_n^{(1)},...,f_n^{(N)}$ are $N$ linearly independent solutions of the linear system
\begin{equation}
        \frac{\partial f_n^{(k)}}{\partial r}=f^{(k)}_{n+1},\qquad 
         \frac{\partial f_n^{(k)}}{\partial s}=-f^{(k)}_{n-1},\qquad 
         1\le k \le N. 
\end{equation}
It was shown in~\cite{BW2010} that a large variety of multi-soliton solutions of the 2DTL can be obtained in this way, including elastic and inelastic, non-resonant, partially resonant and fully resonant ones.

The one- and two-soliton solution recalled in the following subsections are obtained from this construction with $N=1$ and $N=2$ respectively.

Specifically, the one-soliton solution of \eqref{e:Toda2D_0} is obtained by choosing $N=1$ and $f_n^{(1)}=e^{\theta_1}+e^{\theta_2}$, 
where the exponential phases $\theta_1$ and $\theta_2$ given by
\begin{equation}
    \theta_j(n,r,s) = n\log{p_j} + p_j r -\frac{1}{p_j} s + \theta_{j,0},\qquad j=1,2,
\end{equation}
where $0<p_1<p_2$ are positive real parameters and $\theta_{j,0}$ is a
real phase parameter. Introducing the logarithmic parameters $P_i:=\log{p_i}$, so that $p_i=e^{P_i}$ and $P_1<P_2$ and transforming to laboratory frame coordinates (see Appendix \ref{app:B} for the full derivation), the exponential phases become $\tilde \theta_j(n,y,t):=\theta\left ( n, \frac{y+t}{2}, \frac{t-y}{2}\right )$ given by  the following
\begin{equation*}
  \tilde \theta_j(n,y,t)=nP_j+y\cosh{P_j}+t\sinh{P_j}+\theta_{j,0}  
\end{equation*}
where for simplicity we will denote this as $\theta_j(n,y,t)$. The one-soliton solution then takes the form

\begin{equation} \label{SolitonHyperbolic}
V_n(y,t) = \sinh^2 \left[ \tfrac12(P_2-P_1)\right ]
    \sech^2 \left[ \tfrac12 \big( n(P_2-P_1)+y(\cosh{P_2}-\cosh{P_1})+t(\sinh{P_2}-\sinh{P_1}) + t_0+y_0 \big) \right] \,,
\end{equation}
where again for simplicity we overload the notation by setting $V_n(y,t) = V_n(r(y,t),s(y,t))$.
This solution describes a line soliton localized along a line in the $(n,y)$-plane,that translates rigidly as  $t$ increases. 

It is naturally parameterized by two physically meaningful quantities: the amplitude $A$, given by
\begin{equation} \label{Amplitude}
    A_V(P_1,P_2):= \sinh^2 \left( \tfrac12(P_2-P_1)\right),
\end{equation}
and the slope $q$ in the $ny$-plane, given by
\begin{equation} \label{Slope}
q(P_1,P_2):=\frac{\cosh{P_2}-\cosh{P_1}}{P_2-P_1}\,.
\end{equation}
\begin{remark}
In  the limit $m\rightarrow 1$, the genus-one periodic solution degenerates into a  line soliton. Since $V_n=b_n-1$ one has
\begin{equation*}
    A_V= \sup{V}-\inf{V}=(b_{\mathrm{max}}-1)-(b_{\mathrm{min}}-1)=b_{\mathrm{max}}-b_{\mathrm{min}}
\end{equation*}
which is precisely the amplitudes \eqref{eq: SolitonAmplitudeLeftRiemannProblem}  and \eqref{eq: SolitonAmplitudeRightRiemannProblem} used in Section \ref{s:1DTL}. Therefore,  the quantity $A_V$ appearing in \eqref{Amplitude} coincides with the leading-edge soliton amplitude introduced there.  Notice in addition that the soliton solutions have zero background and so $\inf{V}=0$.  
\end{remark}
\begin{remark}
Although we refer to $q$ as the slope parameter, it is in fact an orientation parameter. The corresponding line in the $(n,y)-$ plane has geometric slope $1/q$. We retain the terminology “slope parameter” since $q$ naturally parametrizes the line orientation and coincides with the parameter defining the wedge geometry \eqref{InitialValue2dToda}.
\end{remark}                                                                                            
Note that, in the laboratory coordinates, the slope $q $ can take any real value depending on the choice of $P_1$ and $P_2$, similarly to the KP equation.
This is in contrast  with the parametrization used in \cite{BW2010} where the soliton direction $d_{i,j}= (\log{p_i}-\log{p_j})/(p_i-p_j)$ was always positive for $p_i \neq p_j$. 
Inverting the relation \eqref{Amplitude} and \eqref{Slope} yields (see Appendix \ref{app:B} for details)
\begin{equation}\label{InverseRelation}
    P_1=\gamma - \alpha,\qquad 
    P_2=\alpha + \gamma,
\end{equation}
with 
\begin{equation}\label{e:alphagammafromAq}
\alpha=\arcsinh\big(\sqrt{A_V}\big),\qquad
\gamma = \arcsinh\left( q\frac{\arcsinh\big(\sqrt{A_V}\big)}{\sqrt{A_V}}\right),
\end{equation}
and where again we set $P_2>P_1$ without loss of generality (since \eqref{SolitonHyperbolic} is symmetric under the exchange $P_1 \leftrightarrow P_2$). 
These inversion formulae are a key ingredient for the wedge problem analysis of Section~\ref{s:wedge}.  
The amplitude and slope  of the oblique leading-edge solitons formed in the dynamics are determined from the one-dimensional Whitham theory of Section~\ref{s:Riemann1DTL}, and then converted into the parameters $P_1,P_2$ needed for the two-soliton analysis.

It is important to note in addition that the line solitons described by \eqref{SolitonHyperbolic} propagate in the negative $n$-direction for $t>0$. 
Since the 2DTL equation is invariant under the time-reversal symmetry $t\mapsto -t$, however,
one can write down a line soliton with the same amplitude and slope propagating in the positive $n$-direction for $t>0$. 
More generally, therefore, one may write 
\begin{equation} 
\label{SolitonHyperbolicBid}
V_n(y,t) = \sinh^2 \left[ \tfrac12(P_2-P_1)\right ]
    \sech^2 \left[ \tfrac12 \big( n(P_2-P_1)+y(\cosh{P_2}-\cosh{P_1})+\sigma t(\sinh{P_2}-\sinh{P_1}) + t_0+y_0 \big) \right]\,, 
\end{equation}
with $\sigma = \pm1$, 
where $\sigma=1$ corresponds to propagation in the negative $n$-direction, while $\sigma=-1$ corresponds to propagation in the positive $n$-direction.

\begin{remark}[Extension to arbitrary constant backgrounds]
The family of line soliton solutions~\eqref{SolitonHyperbolicBid} is normalized so that 
$V_n \rightarrow 0$ as $n\to\pm\infty$, corresponding to a zero background state, as indicated above. 
More generally, the 2DTL admits traveling waves on arbitrary constant backgrounds through the following invariance property.
If $\hat{V}_n(y,t)$ is a solution of \eqref{2dTodaLog}, then for every constant $B_0\in \mathbb{R}$,
\begin{equation}\label{SymmetryBackground}
    V_n(y,t)= e^{B_0}\hat{V}_n(e^{B_0/2}y,e^{B_0/2}t) + e^{B_0}-1
\end{equation}
is also an exact solution of \eqref{2dTodaLog}.
\end{remark}
Indeed, note that 
$V_{n+1}-2V_n+V_{n-1}=e^{B_0}( \hat V_{n+1}-2\hat V_n+\hat V_{n-1} )$.
Moreover,
$\log{(1+V_n)}=\log\big( 1+ \hat{V}_n(e^{B_0/2}y,e^{B_0/2}t)\big)+B_0$,
which immediately yields
\begin{equation}
   \left ( \frac{\partial^2}{\partial t^2}-\frac{\partial^2}{\partial y^2}  \right ) \log{(1+V_n)}=e^{B_0} \left( \frac{\partial^2}{\partial \hat t^2}-\frac{\partial^2}{\partial \hat y^2}\right )\log{(1+\hat V_n)},
\end{equation}
where $\hat t=e^{B_0/2}t$ and $\hat y=e^{B_0/2}y$.
Hence both sides of Eq.~\eqref{2dTodaLog} acquire the same multiplicative factor $e^{B_0}$, proving the claim.
Consequently, every line soliton obtained above generates a corresponding family of exact  solutions on the arbitrary  constant background $e^{B_0}-1$.

\subsection{Two-soliton solutions and soliton interactions} \label{s:2Sol}

The two-soliton solutions of the 2DTL are obtained by choosing $N=2$ and $M=4$ in the Casorati determinant construction, where $M$ denotes the number of exponential phases entering  the linear combination 
$f_n^{(k)}= \sum_{m = 1}^M C_{k,m} e^{\theta_m}$ for $k=1,2$,
with four parameters $P_1<P_2<P_3<P_4$ in logarithmic variables
and where the $k\times m$ matrix $C = (C_{k,m})$ is referred to as the coefficient matrix \cite{BW2010}.
In laboratory frame coordinates, the exponential phases are as follows:
\begin{equation} \label{phase}
   \theta_j(n,y,t) = y\cosh{P_j}+nP_j+\sigma t\sinh{P_j}+\theta_{j,0},\qquad j=1,2,3,4\,,
\end{equation}
where again with some abuse of notation we set $\theta_j(n,y,t)=\theta_j(n,r(y,t),s(y,t))$
for $j=1,\dots,4$.
The simplest choice is simply $f_n^{(k)} = e^{\theta_{2k-1}} + e^{\theta_{2k}}$ for $k=1,2$,
to which we refer  as the ordinary solutions, following the terminology of \cite{BW2010};
more general two-soliton solutions are obtained by other choices of the coefficient matrix.
By Theorem~3.8 of \cite{BW2010}, any such solution has exactly two asymptotic line solitons as $n\rightarrow -\infty $(referred to as incoming) and two as $n \rightarrow +\infty$ (outgoing). Each  of these four asymptotic solitons is identified by an index pair $[i,j]$ with $i<j$ and $i,j\in \{1,2,3,4\}$, indicating which two of the four ordered phase parameters $P_1<P_2<P_3<P_4$ determine its amplitude and slope via the formulae of  Section~\ref{s:OneSol}. The index pair is thus a bookkeeping label for a pair of phase parameters, not a name attached to a specific physical soliton: different asymptotic solitons (incoming, outgoing, or emerging from an interaction) can share one of the two indices, and the same physical soliton can carry different index pairs in different regimes. This asymptotic classification is unrelated to the sign $\sigma$ in \eqref{SolitonHyperbolicBid}, which instead labels the direction along the soliton line in which it translates as $t$ increases. 
The elastic two-soliton solutions are those for which the pair of index pairs identifying the incoming one-solitons coincides with that identifying the outgoing  ones; in this case the two solitons pass through each other, acquiring only a finite phase shift, with no permanent change in slope or amplitude. 
Like in the KP equation,  the interaction of elastic two-soliton solutions of the 2DTL is uniquely determined by the degree of overlap between the two index pairs $[i_1,j_1]$ and $[i_2,j_2]$. 

Following~\cite{BW2010}, and taking $i_1 < i_2$ without loss of generality,
three qualitatively distinct scenarios arise:
    \begin{description}
     \item[Resonant interaction:] 
        These solutions manifest when $i_2<j_1<j_2$. 
        In this case, the index pairs are $[1,3]$ and $[2,4]$, both as $n\to -\infty$ and as $n \to\infty$. Unlike the ordinary case, the interaction produces a time-dependent pattern: the two solitons interact via a quartet of Y-shaped junctions, each satisfying Miles' resonance condition, giving rise to a third ``stem'' soliton in each case.
        
        \item[Ordinary interaction:] 
        These solutions manifest when $j_1<i_2$. 
        In this case, the index pairs in the ordered list $P_1<P_2<P_3<P_4$ are $[1,2]$ and $[3,4]$, both as $n\to -\infty$ and as $n \to\infty$. 
        These are traveling wave solutions of the 2DTL.
        At any fixed time $t$, the interaction manifests as an X-shaped pattern in the $ny$-plane. Note that the simplest coefficient matrix choice $f_n^{(k)} = e^{\theta_{2k-1}} + e^{\theta_{2k}}$ for $k=1,2$ introduced before falls precisely into this category, which  is why we refer to both as `ordinary'.
        
        \item[Asymmetric interaction:] 
        These solutions manifest when $j_2<j_1$.
        In this case, the index pairs are $[1,4]$ and $[2,3]$, both as $n\to -\infty$ and as $n \to\infty$. 
        As in the ordinary case, the interaction produces an X-shaped pattern at any fixed time. 
        Unlike the case of ordinary interactions however, in this case the asymptotic line solitons undergo nontrivial position shifts as a result of the interaction.
        This case does not arise in the wedge problem considered in the next section.  
        We refer the reader to \cite{BW2010} for a further characterization.
    \end{description}
\begin{figure}[t!]
    \centering
    \includegraphics[trim= 0 10 0 10,clip,width=\textwidth]{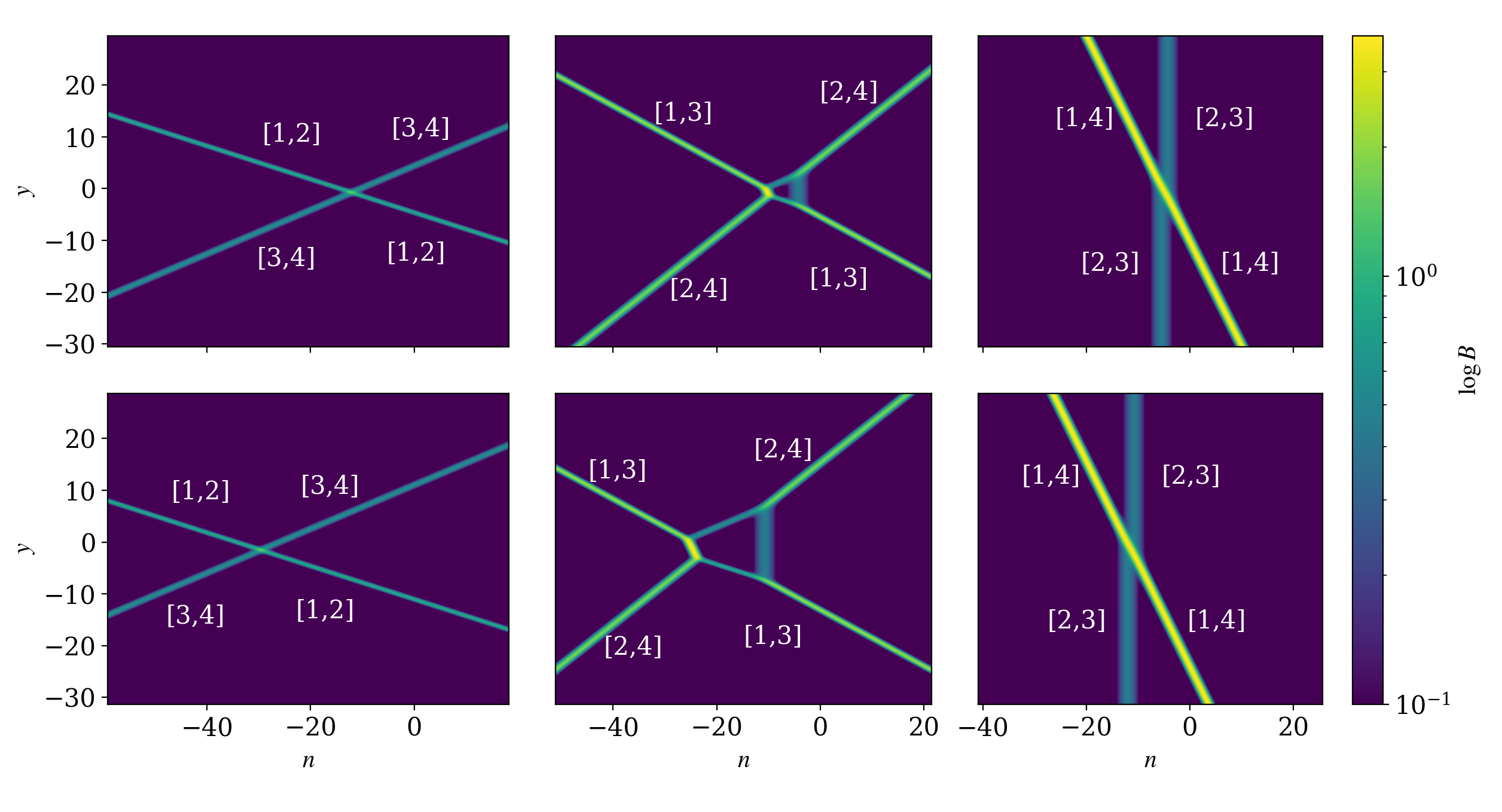}
    \caption{The three types of elastic two-soliton solutions of the 2DTL (index pairs as labeled), with phase parameters $(P_1,P_2,P_3,P_4)=(-2.3,-0.69,0.69,2.71)$, as in \cite{BW2010}. 
    Columns, from left to right: ordinary, resonant, and asymmetric interactions. 
    Rows: earlier and later snapshots (top row: $t = 4$; bottom row: $t = 10$), 
    illustrating the temporal evolution of each solution type. 
    In the resonant interaction (center column), the pair of Y-junctions mediating the stem visibly separates further between the two snapshots.} 
    \label{fig:Interactions}
\end{figure}
Anticipating the results of Section \ref{s:wedge}, we note that resonant interactions of this kind will be responsible for the stem DSW formation in subcritical type II wedges (Section \ref{s:SubII}, Fig. \ref{fig:SubII}), while ordinary interactions will produce the localized peak in supercritical type II wedges (Section \ref{s:SuperII}, Fig. \ref{fig:SuperII}). 
Asymmetric interactions, by contrast, do not arise in the wedge problem considered in this work, as shown in Remark \ref{Obs:Asymmetric} below; we include them here only for completeness of the classification, following \cite{BW2010}. Representative examples of all three interaction types for generic soliton parameters are shown in
Fig.~\ref{fig:Interactions}. We also note that type~I wedges arise from the mirror-image case in which the sign of the parameter $\gamma$ is reversed.
However, since the two oblique DSWs diverge from each other rather than converge in that case, the resulting phenomenology is qualitatively different, as shown in Section \ref{TypeI}.

\subsection{Periodic traveling wave solutions of the two-dimensional Toda lattice}\label{s: PeriodicSolutions2d}

The one-dimensional Toda lattice possesses a family of genus-one periodic traveling waves \eqref{PeriodicSolutions} which form the local building blocks of dispersive shock waves. We now show that the two-dimensional Toda lattice admits a corresponding family of exact oblique periodic traveling waves. We should mention that other classes of periodic solutions of the 2DTL were recently obtained in \cite{duarte2021class, ting2016periodic}.
However, the class of solutions derived here establishes the connection between the finite-gap solutions of the one-dimensional Toda lattice and the two-dimensional wedge problem studied later in Section~\ref{s:wedge}.
 
The genus-one periodic solutions \eqref{PeriodicSolutions} recalled in Section \ref{s:1DTL} are expressed in terms of the elliptic phase $Z_n^{\mathrm{1d}}(t)=2nF(\Xi,m)+\omega_{\mathrm{1d}}t+Z_0$. Following the structure of the one-dimensional finite-gap solution, we extend the elliptic phase to 
\begin{equation} \label{TravelingVariable2d}
    Z_n(y,t)=2nF(\Xi,m)+\ell y+\omega t+Z_0
\end{equation}
where $\ell$ is the transverse wave parameter. Notice that the discrete phase shifts acts as  $Z_{n\pm1}(y,t)=Z_n\pm 2 F(\Xi,m)$ as for $Z_n^{\mathrm{1d}}(t)$. We now consider solutions of Eq. \eqref{2dTodaLog} of the form $V_n(y,t)=V(Z_n(y,t))$ , using the chain rule and substituting into the equation yields
\begin{equation}\label{DDE2d*}
    (\omega^2-\ell^2)\frac{d^2}{dZ^2}\log{(1+V)}=V(Z+2F)-2V(Z)+V(Z-2F)
\end{equation}
Equation \eqref{DDE2d*} is a nonlinear advance-delay differential equation for the traveling-wave profile. Solving it directly appears highly nontrivial.

On the other hand, considering the formulation of the 1DTL in terms of the field $V_n$ defined through the Flaschka variable $a_n$ via
\begin{equation}\label{trasformation}
V_n=4a_n^2-1,
\end{equation}
in terms of which the 1DTL takes the form
\begin{equation} \label{1dTodalog}
    \frac{d^2}{dt^2}\log(1+V_n)= V_{n+1} - 2V_n + V_{n-1}\,,
\end{equation}
which is the $y-$independent reduction of the 2DTL \eqref{e:Toda2D_0} in the laboratory frame coordinates. Substituting in Eq. \eqref{1dTodalog} the one-dimensional phase $Z_n^{\mathrm{1d}}(t)$ one obtains another highly nontrivial nonlinear delay equation
\begin{equation}\label{DDE1d*}
    \omega^2_{\mathrm{1d}}\frac{d^2}{dZ^2}\log{(1+V)}=V(Z+2F)-2V(Z)+V(Z-2F)
\end{equation}
Notice that the reduced equation governing the two-dimensional traveling waves \eqref{DDE2d*} coincides exactly with the one-dimensional traveling wave equation \eqref{DDE1d*} after the identification
\begin{equation}\label{correction*}
   \omega^2-\ell^2  =\omega^2_{\mathrm{1d}}.
\end{equation}
Consequently, every genus one finite gap solution constructed in Section \ref{s:1DTL} generates an exact oblique periodic traveling-wave solution of the two-dimensional Toda lattice. Equivalently, the two-dimensional periodic traveling waves are obtained by replacing the one-dimensional phase $Z_n^{\rm{1d}}(t)$ with the oblique phase $Z_n(y,t)$, while leaving the finite-gap profile and all spectral parameters unchanged. Equation \eqref{trasformation} provides an explicit family of exact periodic traveling-wave solutions of the two-dimensional Toda lattice given by  $a_n(t) = a(Z_n(y,t))$ and $V_n(y,t) = 4a_n^2-1$,
with 
$Z_n(y,t) = 2nF(\Xi,m)+\ell y+\omega t+Z_0$. The situation is analogous to what happens for the KP equation (e.g., see \cite{RSPA2017v473p20160695}), except for the fact that, in that case, the correction \eqref{correction*}  affects $\omega$ rather than $\omega^2$, owing to the fact that KP is first-order in time while 2DTL is second-order.

Since the numerical simulations throughout this work are performed in the strain variable $B_n$, it is useful to express the above family in those coordinates as well, specifically one has
\begin{equation*}
    B_n(y,t)=\log{\left (4a_n^2(Z_n(y,t))\right )}
\end{equation*}
\begin{remark}
  The previous derivation shows that every oblique traveling wave of the 2DTL satisfies exactly the same nonlinear advance-delay equation as the one-dimensional traveling wave reduction of the 1DTL. Consequently, the family of oblique periodic traveling waves of the two-dimensional Toda lattice is precisely the genus-one finite-gap family constructed in Section \ref{s:1DTL}, with the only modification being the extension of the elliptic phase from $Z_n^{\mathrm{1d}}(t)$ to $Z_n(y,t)$. This observation provides the mathematical justification for employing the one-dimensional Whitham theory in the study of the two-dimensional wedge problem presented in  the next section. Since the Whitham modulation equations are derived from this family of periodic traveling waves, and this family is unchanged by the above reduction, the modulation of each oblique branch of the wedge is governed locally by the same one-dimensional periodic solutions derived in Section \ref{s:1DTL}. The transverse direction influences only the evolution of the phase through the relation \eqref{correction*} while leaving the underlying periodic profile unchanged.
\end{remark}
\begin{remark}
    The elliptic $Z_n$ introduced above is the natural variable arising from the finite-gap construction of the one-dimensional Toda lattice \cite{Teschl} and is therefore not the standard physical phase. Since the periodic solutions \eqref{PeriodicSolutions} depend on $\sn^2(Z,m)$, whose period is two times  the complete elliptic integral of the first kind $2K(m)$, one may introduce the rescaled phase
    \begin{equation*}
        \Phi=\frac{\pi}{K(m)}Z
    \end{equation*}
    which is $2\pi-$periodic. In  terms of $\Phi$, the traveling-wave variable assumes the well-known form
    \begin{equation*}
        \Phi_n(y,t)=k_1n+k_2 y+\Omega t+\Phi_0
    \end{equation*}
where $k_1=\frac{2\pi F(\Xi,m)}{K(m)}$ is the lattice wavenumber, $k_2=\frac{\pi}{K(m)}\ell$ is the transverse wavenumber, and $\Omega=\frac{\pi}{K(m)}\omega$ is the temporal frequency. Rewriting the genus-one finite-gap solutions of Section \ref{s:1DTL} in terms of the physical phase $\Phi$ therefore provides an equivalent representation of the same family of exact oblique periodic traveling waves of the 2DTL, now expressed in terms of the physical wave parameters $(k_1,k_2,\Omega)$. Throughout this work, however, we retain the elliptic phase $Z$, since it is the natural coordinate for the finite-gap formulation.
\end{remark}

\section{Wedge problems for the two-dimensional Toda lattice}
\label{s:wedge}

We are now ready to study the main topic of this work, namely the dynamics produced by equation~\eqref{eq_Toda1} subject to wedge-type initial data of the form \eqref{InitialValue2dToda}, 
where $B_->B_+$ are constants and $q\in \mathbb{R}$ is the slope parameter of the wedge.

Throughout this section we make use of two equivalent dependent variables. The wedge problem is formulated in terms of the strain variable $B$, since the initial data constitute a Riemann problem in $B$, and this is also the natural variable in the long-wave reduction to the KP equation explored in the next section. On the other hand, the exact one- and two-soliton solutions of the 2DTL, presented in the previous section, available in the literature are naturally expressed in terms of the variable $V=e^B-1$. We therefore employ the $V-$ formulation to exploit the available exact soliton theory, while all physical predictions are subsequently translated back to the strain variable $B$ for comparison with the wedge dynamics.

\subsection{General set-up}

As illustrated in Figure \ref{fig:wedge}, the sign of $q$ determines the orientation of the wedge: for type I wedges ($q>0$) the region of higher amplitude $B_-$ forms a rightward-pointing wedge in the $(n,y)$-plane, while for type II wedges ($q<0$) it forms a leftward-pointing wedge.

Since \eqref{eq_Toda1} is second-order in time, specifying $B_n(y,0)$ alone does not uniquely determine the solution.
One must also prescribe a second initial condition $\dot{B}_n(y,0) = \varphi_n(y)$. 
The choice of $\varphi_n(y)$ is therefore part of the problem formulation and has a direct dynamical consequence. To see this, observe that the jump in $B_n(y,0)$ across the wedge boundary $n+q|y|=0$ acts locally as a Riemann problem for the one-dimensional Toda lattice in the Flaschka variables $(a_n,b_n)$ with $b_n=e^{B_n}$ and $\dot{B}_n=a_{n+1}-a_{n}$. The well-posedness of this local Riemann problem follows from the strict hyperbolicity of the Whitham system \eqref{WhitamTodaGenusg}: the characteristic speeds $s_i$ are real and distinct for distinct branch points, guaranteeing the existence and uniqueness of the solution for all  time in Eq. \eqref{WhitamTodaGenusg}. 
However, within this framework, the choice of $\varphi_n(y)$ selects which DSWs are excited. 
In order to isolate the right-going DSW, we choose $\varphi_n(y)$ so that $\lambda_2=a_n(0)+2\sqrt{b_n(0)}$ remains constant across the jump.  
Setting $a_-=0$ the condition $\lambda_2^-=\lambda_2^+$ leads to:
\begin{equation} \label{velocity}
    a_+=2 \big(e^{{B_-}/{2}}-e^{{B_+}/{2}} \big)\,,
\end{equation}
which leads to a well-posed initial value problem for \eqref{e:Toda2D_0} with initial velocity given by 
\begin{equation}\label{InitialVelocity}
    \dot{B}_n(y,0)=2 \big(e^{{B_-}/{2}}-e^{{B_+}/{2}} \big)\chi_{\{n=\lfloor -q|y| \rfloor\}}
\end{equation}
where $\lfloor \cdot \rfloor $ denotes the floor function.

Before proceeding further, it is important to observe that the wedge problem \eqref{InitialValue2dToda} with arbitrary constant background $B_+$ can always be reduced to a wedge with zero background considered throughout the  numerical simulations. 
Indeed, the invariance~\eqref{SymmetryBackground} established in Section~\ref{s:2DTL_background} implies that any solution of the 2DTL \eqref{eq_Toda1} generates a one-parameter family of solutions via the map
\begin{equation}
    B_n(y,t)\mapsto B_n\left (e^{B_0/2}y,e^{B_0/2}t \right )+B_0\,.
\end{equation} 
Therefore, a wedge \eqref{InitialValue2dToda} can be transformed  into an equivalent normalized (i.e., zero background) wedge by introducing rescaled quantities
\begin{equation}
    \hat{B}_n=B_n-B_+ \qquad \hat y=e^{B_+/2}y, \qquad \hat t=e^{B_+/2}t
\end{equation}
The transformed problem satisfies the same Toda equations and it is of the form:
\begin{equation} 
    \hat B_n(y,0)=\begin{cases}
       \Delta , &n + q_{\mathrm{eff}}\,|\hat y| \le 0,\\
        0, &n + q_{\mathrm{eff}}\,|\hat y| > 0,
    \end{cases}
    \end{equation}
with 
\begin{equation}
\Delta=B_--B_+,\qquad q_{\mathrm{eff}}=e^{-B_+/2}q\,.
\end{equation}
Hence, we can always reduce ourselves to working exclusively with a normalized wedge problem with zero background without loss of generality, i.e., to taking $B_+=0$. 
We will do so throughout the rest of this work,
which will allow us to more easily take advantage of the analysis of the exact multi-soliton solutions of the 2DTL that were discussed in section~\ref{s:2DTL_background}.

\begin{remark}
The choice of imposing $\lambda_2$ constant isolates the right-moving dispersive shock. Indeed, let $b_\pm =e^{B_\pm}$ and choose the initial data so that $\lambda_2^+=\lambda_2^-$ that is Eq. \eqref{velocity} with the choice $a_-=0$. Substituting  these values into the amplitude Eqs. \eqref{eq: SolitonAmplitudeLeftRiemannProblem} and \eqref{eq: SolitonAmplitudeRightRiemannProblem} one has $A_b^L=0$ and $A_b^R>0$ as discussed later.
\end{remark}
Moreover, the previous section confirms the fact that the amplitude of the leading-edge of each oblique DSW is determined entirely by the one-dimensional Whitham theory of Section~\ref{s:Riemann1DTL}. 
Specifically, applying \eqref{eq: SolitonAmplitudeRightRiemannProblem} to the initial data \eqref{InitialValue2dToda} yields the following estimate for the amplitude of the leading edge solution in the oblique DSW:
\begin{equation}\label{e:2DTL_obliqueleadingedgeamplitude}
  A_{B,\mathrm{obl}}=\log{\left( \frac{b^R_{\mathrm{max}}}{b^R_{\mathrm{min}}}\right)}=2\log{ \left (2e^{\frac12\Delta}-1 \right )}\,,
\end{equation}
where $\Delta = B_--B_+>0$ is the size of the jump.
\begin{remark}\label{obs: AmplitudeOblique}
Unlike what happens for the KP equation in the continuum limit, the 2DTL does not possess a Galilean/shift invariance. 
Nevertheless, \eqref{e:2DTL_obliqueleadingedgeamplitude} shows that,
like in the KP equation, the amplitude of the leading-edge soliton in a DSW of the 2DTL depends only on the size of the jump $\Delta$, and it is independent of the background.
Note also that, similarly to the KP equation, the amplitude is also independent of the orientation of the jump. 
\end{remark}

\begin{remark}
The two-dimensional nature of the problem manifests when the two oblique DSWs interact around the symmetry axis $y=0$. The dynamics of this interaction is not captured by the one-dimensional theory, and indeed it depends crucially on the geometry of the wedge, namely on the slope parameter~$q_{\mathrm{eff}}$. 
In the wedge problem for the KP equation, the resulting dynamics was studied and classified in \cite{BBH_PRL2025} using suitable reductions of the Whitham modulation equations for the KP equation,
which had been derived in \cite{RSPA2017v473p20160695}.
No Whitham modulation equations are available for the 2DTL to the best of our knowledge, although their derivation is a particularly 
intriguing question for future work.
Nonetheless, we next show that the dynamical patterns can be studied and classified to the same 
level of effectiveness by looking at the interactions of the leading-edge solitons of the two oblique DSWs.
Indeed, the two oblique leading-edge solitons, having the same amplitude $A_{B,\mathrm{obl}}$ but opposite slopes $q_{\mathrm{eff}}$ and $-q_{\mathrm{eff}}$, can interact in two inequivalent ways depending critically on the relative ordering between the value of $|q_{\mathrm{eff}}|$ and the critical value $q_{\mathrm{cr}}$. 
Specifically, we show below that the following two scenarios can arise:
\begin{description}
    \item[Subcritical regime  ($|q_{\mathrm{eff}}|<q_{\mathrm{cr}}$):] 
    The interaction between the leading-edge solitons is resonant, producing an expanding vertical region called a stem, which generates a third, vertical ``stem DSW''. 
    This scenario is studied in Section \ref{s:SubII}.
    \item[Supercritical regime ($|q_{\mathrm{eff}}|>q_{\mathrm{cr}}$):] 
    The interaction between the leading-edge solitons is ordinary; the two solitons pass through each other without forming a stem, and the interaction produces instead a localized peak of high amplitude, characterized analytically in  Section \ref{s:SuperII}. 
\end{description}
\end{remark}

The critical slope $q_{\mathrm{cr}}$ that separates the above two regimes is determined by identifying the resonance condition on the phase parameters of the two leading-edge solitons.
Recall that these solitons have equal amplitudes and opposite slopes.

Let $\alpha$ and $\gamma$ be the parameters associated with the leading-edge soliton in the upper half plane via \eqref{e:alphagammafromAq} with $A_V=A_{V,\mathrm{obl}}$, with 
\begin{equation}
A_{V,\mathrm{obl}} = 4 e^{\Delta/2}(e^{\Delta/2}-1)\,,
\end{equation}
and where $A_{V,\mathrm{obl}}=e^{A_{B,\mathrm{obl}}}-1$. .
The four phase parameters for the two solitons are then given by the set 
$\{-\gamma-\alpha,\alpha-\gamma,\gamma-\alpha,\gamma+\alpha\}$.
It is convenient to denote the well-ordered set of phase parameters for the two-soliton solutions by 
\begin{equation}
(P_1,P_2,P_3,P_4):=\mathop{\rm ord}\{-\gamma-\alpha,\alpha-\gamma,\gamma-\alpha,\gamma+\alpha\}\,.
\label{e:P_ord}
\end{equation}
(As in the case of the KP equation, the ordering of the solitons parameters $P_1,\dots,P_4$ is needed in the formalism presented of Section~\ref{s:2Sol} in order to obtain non-singular soliton solutions.)
The precise ordering in our setting depends on the sign of $\gamma$ and the relative magnitude of $\gamma$ and $\alpha$ as discussed below. 
(Note that $\alpha$ is always positive, while $\gamma$ can be either positive or negative.)
Using the expression \eqref{e:alphagammafromAq} for the soliton parameters $P_2$ and $P_3$ derived in Section \ref{s:OneSol} the resonance condition is obtained by imposing $P_2=P_3$, namely
\begin{equation*}
    \arcsinh(\sqrt{A_{V,\mathrm{obl}}})=\arcsinh\left( |q_{\mathrm{eff}}| \frac{\arcsinh(\sqrt{A_{V,\mathrm{obl}}})}{\sqrt{A_{V,\mathrm{obl}}}}\right),
\end{equation*}
which in turn implies that the critical slope is:
\begin{equation} \label{Eq:Res_qA}
   q_{\mathrm{cr}} = \frac{A_{V,\mathrm{obl}}}{\arcsinh(\sqrt{A_{V,\mathrm{obl}}})}\,.
\end{equation}

\begin{remark}\label{obs: alphaGamma}
  Note that both $\alpha$ and $\gamma$ are explicit functions of $q_{\mathrm{eff}}$ and of the jump size $\Delta$.  
\end{remark}
The precise ordering of the parameters in \eqref{e:P_ord} and the resulting dynamical regime for each value of $q_{\mathrm{eff}}$ are summarized in the following:

\begin{remark} \label{Obs:Cases}
There are precisely four cases, which correspond exactly to the four dynamical regimes analyzed later in our work. The four scenarios arise by combining the two possible wedge orientations (type I and type II), determined by the sign of $q_{\mathrm{eff}}$, with the two possibilities determined by whether $|q_{\mathrm{eff}}|$ is below or above the critical value $q_{\mathrm{cr}}$.
\begin{description}
\item[\textbf{Case 1:}] $0<\gamma<\alpha$. 
    In this case $(P_1,P_2,P_3,P_4)=(-\gamma-\alpha,\gamma-\alpha,\alpha-\gamma,\alpha+\gamma)$, 
    which leads to subcritical type~II wedges. 
    The index pairs are given by $[1,3]$ for the soliton of slope $q$ and $[2,4]$ for the soliton of slope~$-q$.
\item[\textbf{Case 2:}] $0<\alpha<\gamma$. 
    In this case $(P_1,P_2,P_3,P_4)=(-\gamma-\alpha,\alpha-\gamma,\gamma-\alpha,\alpha+\gamma)$, 
    which leads to supercritical type~II wedges. 
    The index pairs are given by $[1,2]$ for the soliton of slope $q$ and $[3,4]$ for the soliton of slope~$-q$.
\item[\textbf{Case 3:}] $-\alpha<\gamma<0$. 
    In this case $(P_1,P_2,P_3,P_4)=(\gamma-\alpha,-\gamma-\alpha,\alpha+\gamma,\alpha-\gamma)$, 
    which leads to subcritical type~I wedges. 
    The index pairs are given by $[1,3]$ for the soliton of slope $-q$ and $[2,4]$ for the soliton of slope~$q$.
\item[\textbf{Case 4:}] $\gamma<-\alpha<0$. 
    In this case $(P_1,P_2,P_3,P_4)=(\gamma - \alpha, \gamma + \alpha, - \gamma - \alpha, \alpha - \gamma)$, which leads to supercritical type~I wedges. 
    The index pairs are given by $[1,2]$ for the soliton of slope $-q$ and $[3,4]$ for the soliton of slope~$q$.
\end{description}
\end{remark}

In the first two cases, the resonance condition is given by $\gamma - \alpha = \alpha - \gamma$, which gives $\gamma=\alpha$.
In the last two cases, the resonance condition is given by $-\gamma - \alpha = \alpha + \gamma$, which gives $\gamma=- \alpha$.
In all four cases, the condition $P_2=P_3$ leads to $|\gamma|=\alpha$, which in turn yields expression \eqref{Eq:Res_qA}. Notice that the critical value $q_{\mathrm{cr}}$ is positive, since $A_{V,\mathrm{obl}}>0$. The distinction between Type I and Type II wedges therefore does not affect the magnitude of the critical slope, but only its sign: Type I corresponds to $q_{\mathrm{eff}}=q_{\mathrm{cr}}$, whereas Type II corresponds to $q_{\mathrm{eff}}=-q_{\mathrm{cr}}$. Hence, the transition between ordinary and resonant interaction is determined by the threshold $|q_{\mathrm{eff}}|=q_{\mathrm{cr}}$.

\begin{remark}\label{Obs:Asymmetric}
Note that the asymmetric interaction described in Section \ref{s:2Sol} does not arise in any of the four cases above.
\end{remark}


In the following sections we analyze each of the four dynamical regimes identified in Remark~\ref{Obs:Cases}, providing both analytical predictions and their numerical confirmation.
All simulations presented in the following sections are performed with parameters $B_+=0$ and $B_-=0.5$, which corresponds via \eqref{Eq:Res_qA} to a critical value $|q_{\mathrm{cr}}|\approx 1.42884$. {Since $B_+=0$, we have $q_{\mathrm{eff}}=q$ so for notational simplicity we henceforth drop the subscript ``eff".}  The initial profile is smoothed by a hyperbolic tangent transition to avoid spurious small-scale oscillations. 
(Note that this smoothing does not affect the analytical predictions, as the Whitham modulation theory determining the leading-edge soliton amplitude and the critical slope $q_{\mathrm{cr}}$ depends only on the asymptotic states $B_\pm$, not on the fine structure of the transition layer connecting them. Further numerical details are given in Appendix \ref{app:D}.)

The initial velocity $\dot{B}_n(y,0)$ is prescribed using Eq. \eqref{velocity}  to isolate the right-going dynamics. Time integration is performed using the symplectic Verlet scheme in a domain of $N=2400$ lattice sites in $n$ and $M=6000$ points in $y$ with spacing $\Delta y=0.25$, up to $t=300$; full numerical details are given in Appendix \ref{app:D}.

\subsection{Subcritical type II wedges} 
\label{s:SubII}

\begin{figure}[]
    \centering
    \includegraphics[width=\textwidth]{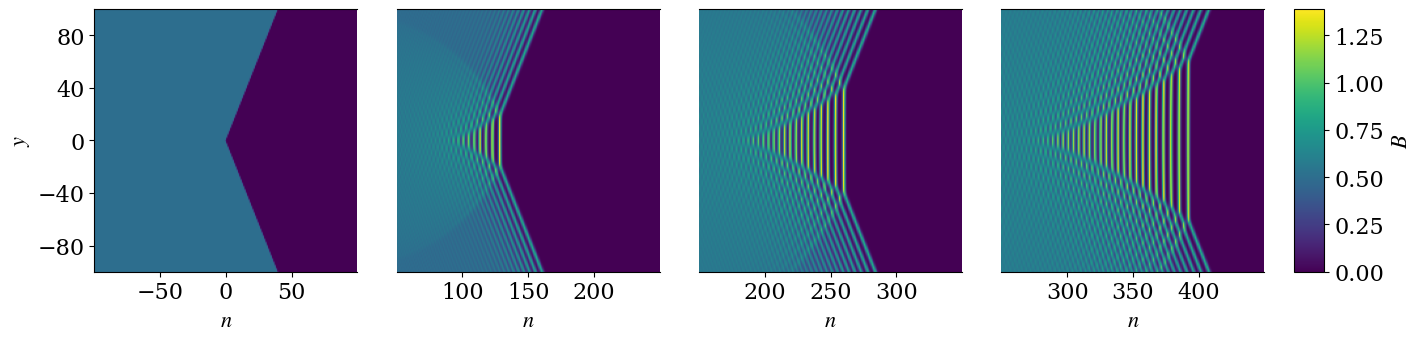}
    \caption{Density plots illustrating the evolution of a subcritical type II wedge with $q = -0.4$. The snapshots are at $t = 0,\ 100,\ 200,\ 300$. Note the formation of the vertical stem and its growth in $y$ as time progresses.}
    \label{fig:SubII}
\end{figure}

In the subcritical type II regime ($q<0$ and $|q|<q_{\mathrm{cr}}$, corresponding to Case~1 of Remark~\ref{Obs:Cases}) the two oblique dispersive shock waves originating from the wedge interact resonantly along $y=0$, generating a third, vertical ``stem DSW'' of higher amplitude. 
The time evolution is illustrated in Fig.~\ref{fig:SubII}, which shows four snapshots at $t=0,100,200,300$: the initial wedge profile gives rise to two symmetrically propagating DSWs whose interaction generates a third vertical DSW, the stem. The stems spatial extent grows linearly in time. 
This phenomenon is a direct consequence of the resonant interaction described in  Section~\ref{s:2Sol}.
The leading-edge solitons of the two oblique DSWs are described by the parameters
\begin{equation*}
(P_1,P_2,P_3,P_4)=(-\gamma-\alpha, \gamma-\alpha, \alpha - \gamma, \alpha + \gamma),
\end{equation*}
where $0<\gamma<\alpha$, and correspond to index pairs $[1,3]$ (soliton of slope $q$) and $[2,4]$ (soliton of slope $-q$). 
Following the classification of resonant two-soliton solutions of the 2DTL in \cite{BW2010} the stem soliton generated from this resonant interaction is the leading-edge soliton of the vertical DSW, described by the parameters  $P_1$ and $P_4$.
Its slope and amplitude are, respectively,
\bse
\label{parameters_stem}
\begin{align} 
  q_{\mathrm{stem}} &= \frac{\cosh{P_4}-\cosh{P_1}}{P_4-P_1}=0, 
  \\
  A_{B,\mathrm{stem}} &= \log\left(1+\sinh^2 \left( \tfrac12(P_4-P_1)\right)\right) = 
     \log \big(1+\sinh^2 (\alpha+\gamma)\big).
\end{align}
\ese
Since $\alpha$ and $\gamma$ depend explicitly on $q$ and $\Delta$, as remarked in \eqref{obs: alphaGamma}, so does $A_{B,\mathrm{stem}}$,
although we retain the compact notation in $\alpha,\gamma$ for brevity.

\begin{remark}
The fact that $q_{\mathrm{stem}}=0$ confirms that indeed the stem DSW that is formed by the dynamics is vertical.
\end{remark}
The vertical extent of the stem DSW grows linearly in time, while the DSW also expands and propagates  horizontally. 
These expansion rates can be tracked by noticing that the bottom edge $(n_-(t),y_-(t))$ of the stem is given by the intersection of the three solitons $[1,2], \ [1,4]$ and $[2,4]$. This triple intersection is equivalent to imposing $\theta_1=\theta_2=\theta_4$, which yields a linear system in $n$ and $y$.
The solution of this linear system (computed in Appendix \ref{app:C}) yields
\bse
\begin{equation}
\label{upper vertex}
    y_-(t) = - \frac{\alpha\coth{\alpha-\gamma\coth{\gamma}}}{\alpha+\gamma} t+y_-^0,\qquad
    n_-(t) = \frac{\sinh{(\alpha+\gamma)}}{\alpha+\gamma}t+n_-^0,
\end{equation}
with 
\be
y_-^0 = - \frac{(\theta_{1,0}-\theta_{4,0})\alpha+\theta_{1,0}-\theta_{2,0}}{2\sinh\gamma\sinh\alpha},\qquad
n_-^0 = \tfrac12(\theta_{4,0}-\theta_{1,0}).
\ee
\ese
Similarly, the upper edge $(n_+(t),y_+(t))$ is given by the intersection of the solitons $[2,3], \ [2,4]$ and $ [3,4]$. The associated linear system yields
\bse
\begin{equation} 
\label{lower vertex}
    y_+(t) = + \frac{\alpha \coth{\alpha}-\gamma\coth{\gamma}}{\alpha+\gamma}t+y_{+}^0,\qquad
    n_+(t) = \frac{\sinh{(\alpha+\gamma)}}{\alpha+\gamma}t+n_+^0,
\end{equation}
with 
\be
y_+^0=\frac{(\theta_{3,0}-\theta_{4,0})\alpha+\theta_{3,0}-\theta_{2,0}}{2\sinh\gamma\sinh\alpha},\qquad
n_+^0 = \tfrac12(\theta_{4,0}-\theta_{3,0}).
\ee
\ese

\begin{remark}
The constants $\theta_{j,0}$ contribute an overall offset for each leg in the two-soliton solution. 
In general, these constants are arbitrary, and only affect the position of the interaction pattern. 
However, for the symmetric wedge ICs considered here, the reflection symmetry $y\mapsto -y$ exchanges the two oblique leading solitons, namely $[1,3]$ and $[2,4]$. 
At the level of the phase constants, this symmetry results in the further conditions
$\theta_{1,0}=\theta_{2,0}$ and $\theta_{3,0}=\theta_{4,0}$.
The remaining relative phase is fixed by requiring that the resonant interaction is generated at the wedge tip, located at $(n,y)=(0,0)$.
Equivalently, the condition
$(n_\pm^0,y_\pm^0)=(0,0)$
forces all phase constants to coincide:
$\theta_{1,0}=\theta_{2,0}=\theta_{3,0}=\theta_{4,0}$.
A common value is irrelevant because it produces only a global multiplicative factor in the tau-function. 
Therefore, without loss of generality we set $\theta_{j,0}=0$ for all $j=1,...,4$.
\end{remark}

Combining \eqref{upper vertex} and~\eqref{lower vertex} yields the length $L(t):=y_+(t)-y_-(t)$ and the horizontal velocity $c_{\mathrm{stem}}$ of the leading-edge soliton in the stem DSW respectively as
\bse
\label{theoretical}
\begin{align} 
    L(t) &= 2\frac{\alpha\coth\alpha - \gamma\coth\gamma}{\alpha+\gamma}t,
    \\
    c_{\mathrm{stem}} &= \frac{\sinh{(\alpha+\gamma)}}{\alpha+\gamma}, 
\end{align}
\ese
where we used the fact that $L(0)=y_+^0-y_-^0=0$.

\begin{figure}[t!]
    \centerline{
        \includegraphics[width=0.48\linewidth]{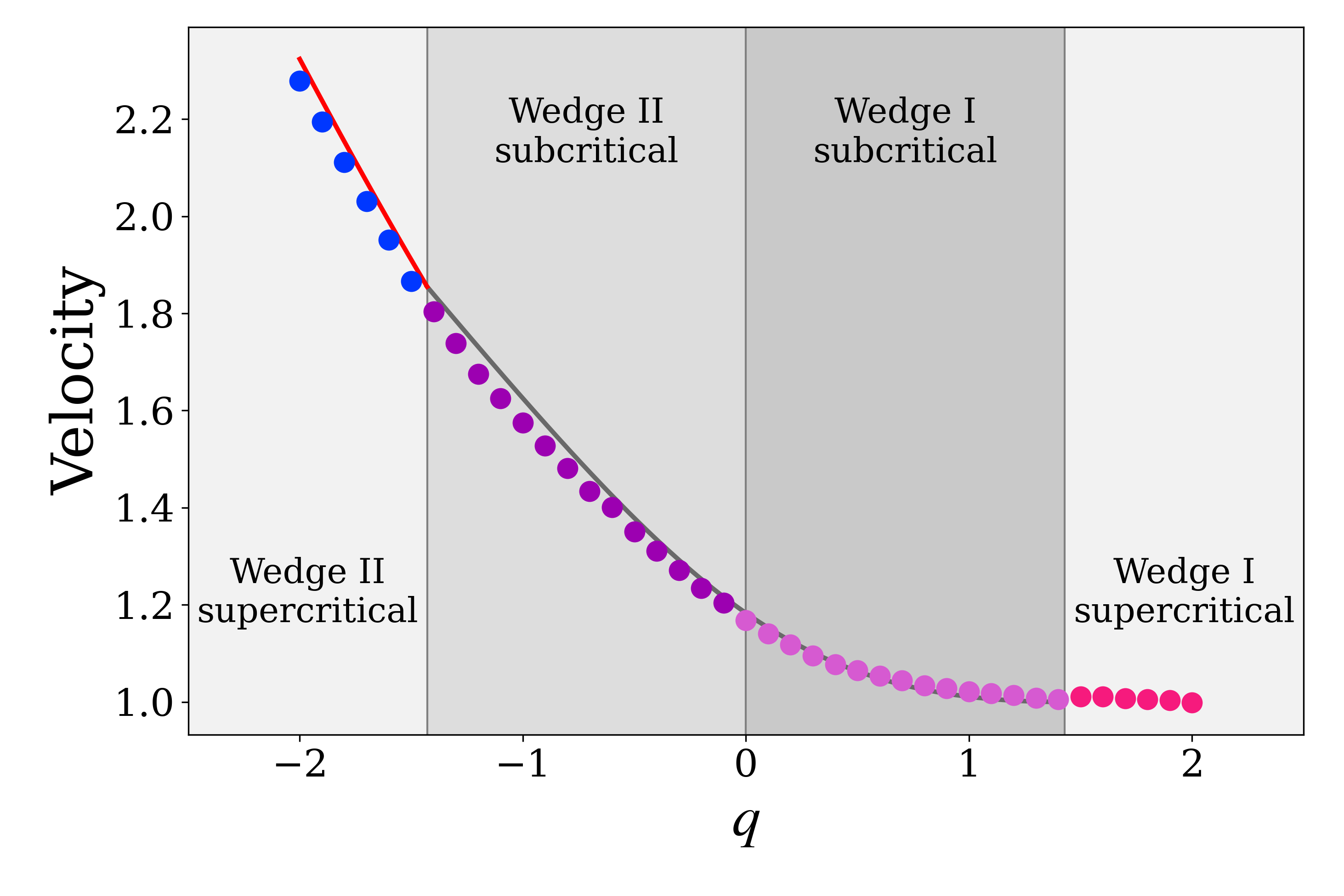}
        \includegraphics[width=0.48\linewidth]{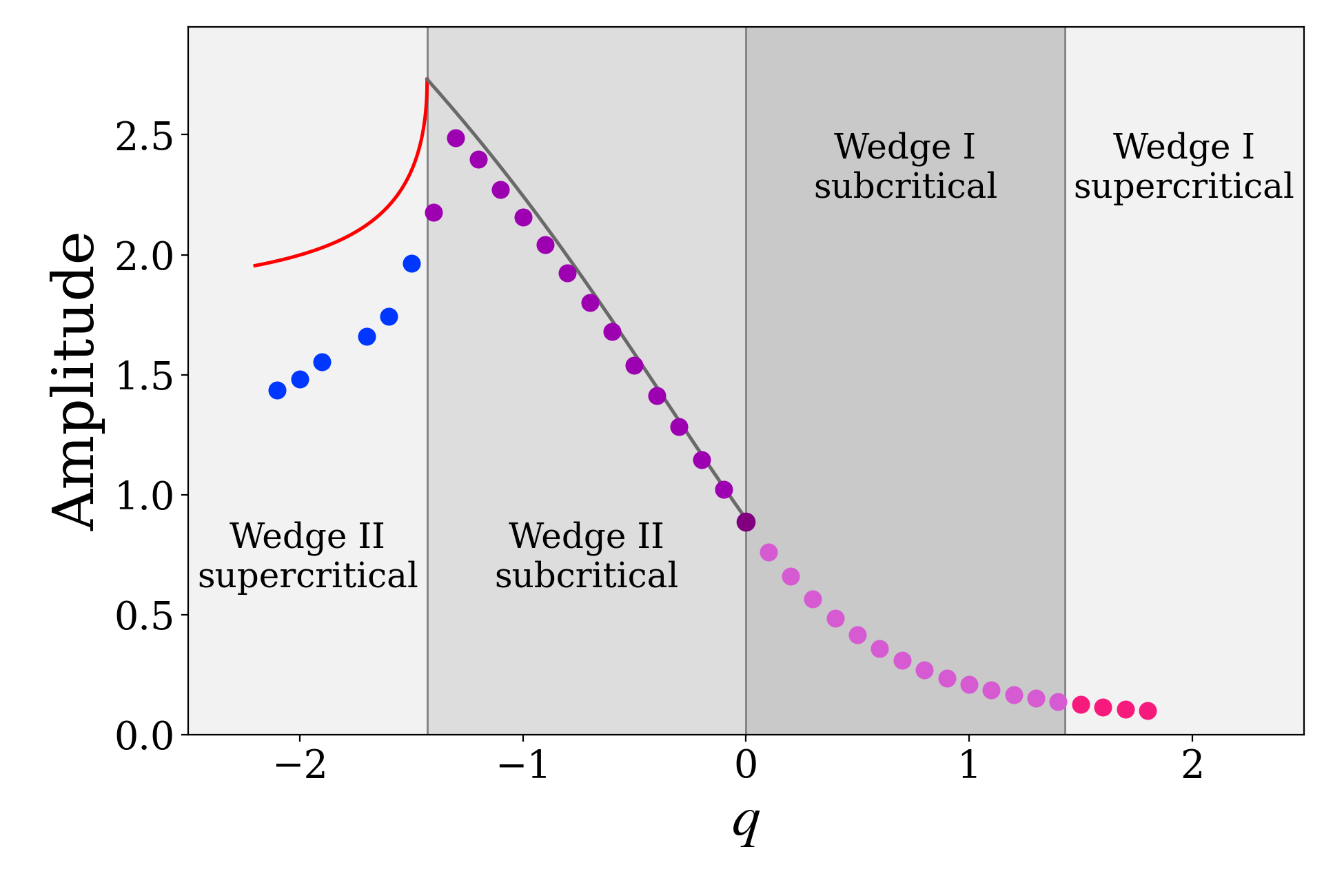}
        \qquad}
\kern-\medskipamount
\caption{Left: Numerically measured horizontal velocity of the leading-edge as a function of $q$: $c_{\mathrm{stem}}$ 
for subcritical type~II wedges (purple dots) and subcritical type~I wedges (pink dots)
as a function of $q$ (dots) compared to the theoretical prediction from \eqref{theoretical} (solid gray curve); $c_{\mathrm{peak}}$ for supercritical type~II  wedges (blue dots) and supercritical type I wedges (magenta dots). Supercritical type ~II wedges are compared with the theoretical prediction from\eqref{peak} (solid red curve).
Right: Numerically measured maximum amplitude at $y=0$ as a function of $q$ for 
subcritical type~II wedges (purple dots) and supercritical type~II wedges (blue dots)
compared with the theoretical prediction $A_{B,\mathrm{stem}}$ from \eqref{parameters_stem} (solid gray line)
and $A_{B,\mathrm{peak}}$ from \eqref{peak} (red solid line).
}
\label{fig:Amplitudes}
\kern1.4\bigskipamount
    \centering
    \begin{subfigure}{0.49\textwidth}
        \includegraphics[width=\linewidth]{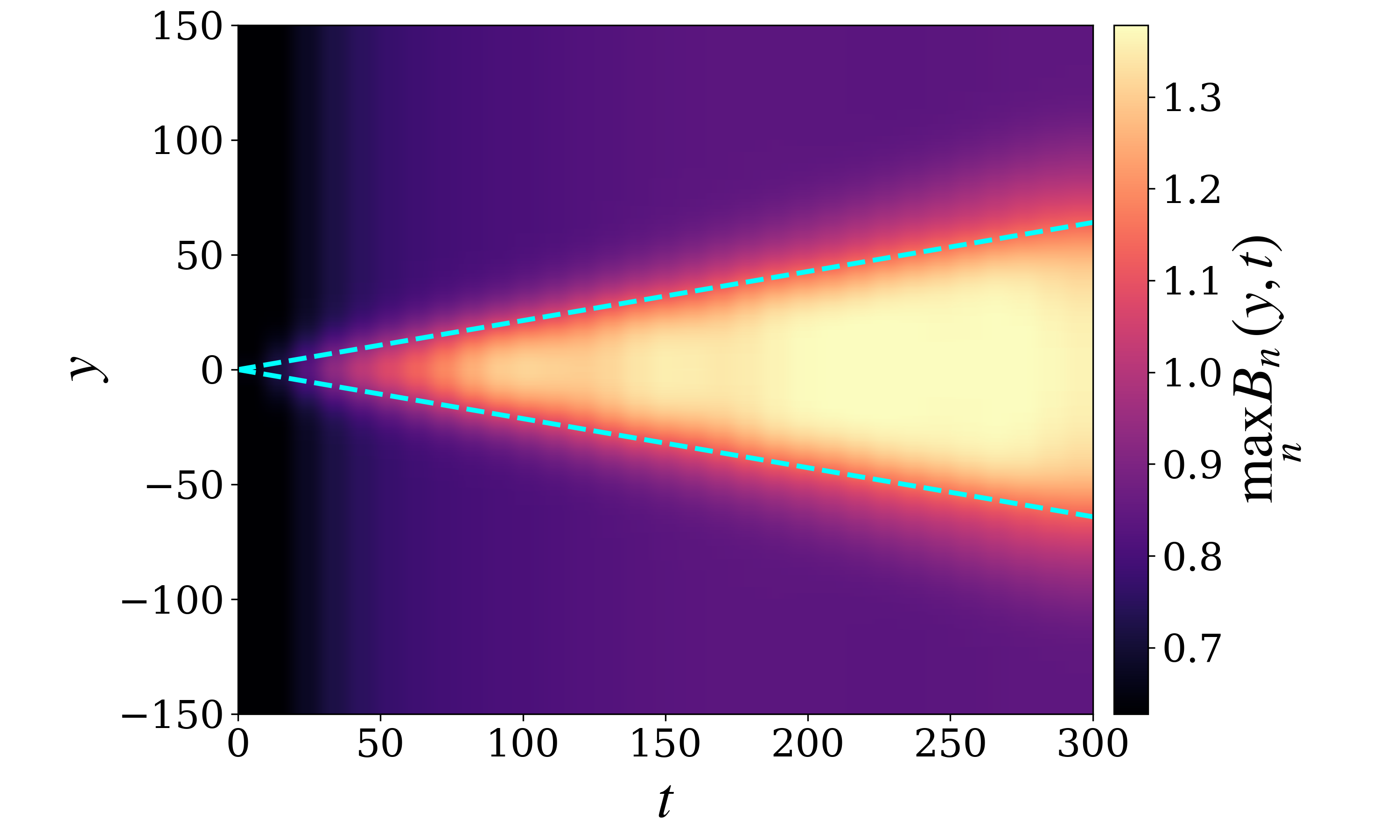}
    \end{subfigure}
    \hfill
    \begin{subfigure}{0.49\textwidth}
        \includegraphics[width=\linewidth]{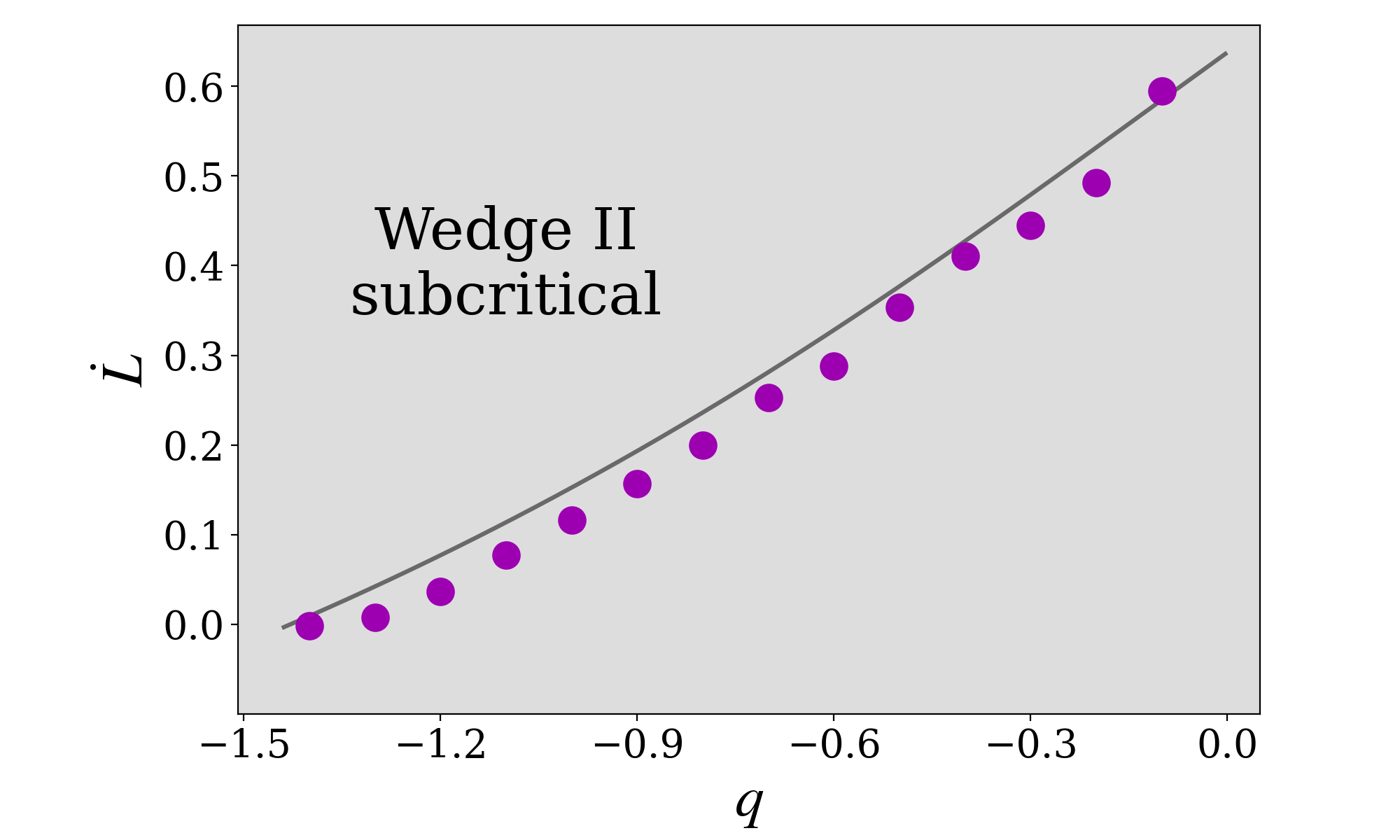}
    \end{subfigure}
    \caption{Left: Space-time density plot of $\max_nB(y,t)$ for the subcritical type II wedge with $q=-0.4$,  smoothed via spline interpolation with Gaussian filters. The cyan dashed lines show the analytical prediction for the stem $y_\pm (t)$ from equations \eqref{upper vertex} and \eqref{lower vertex}. Right: Numerically measured vertical expansion rated $\dot{L}$ as a function of $q$ (purple dots) compared with the theoretical prediction \eqref{theoretical} (solid curve) for all subcritical type II wedges $-q_{\mathrm{cr}}<q<0$.
   } 
    \label{fig:max}
\end{figure}

\begin{remark}
In the subcritical type II regime, i.e., for $0< \gamma < \alpha$, one has $\dot{L}>0$ confirming that the stem expands vertically in time, as one can observe in Fig.~\ref{fig:SubII}. 
Moreover, notice also that, in the limit $\gamma\to\alpha$, the length $L$ tends to $L_0$, which in the case of our set-up is zero, consistently with the transition from subcritical to supercritical regime where no stem formation is observed.
\end{remark}
\begin{remark}
One has $c_{\mathrm{stem}}>0$, in agreement on the rightward propagating nature of the dynamics for the initial conditions considered here. Note that, unlike the KP equation, the 2DTL is invariant under parity, so this rightward propagation is a consequence of the specific choice of initial data \ref{InitialVelocity} which isolates the right-going branch.
\end{remark}

All theoretical predictions are confirmed by numerical simulations, as shown in Figs.~\ref{fig:Amplitudes} and~\ref{fig:max}. In particular, Fig.~\ref{fig:Amplitudes} shows the numerically measured horizontal stem velocity $c_{\mathrm{stem}}$ (left panel) and the  maximum  amplitude $A_{B,\mathrm{stem}}$ 
as a function of $q$, both compared to the theoretical predictions $\eqref{theoretical}$ and $\eqref{parameters_stem}$, finding good agreement in the subcritical regime $-q_{\mathrm{cr}}<q<0$.  Figure \ref{fig:Amplitudes} provides additional  confirmation: the left panel shows a space-time density plot for $\max_{n} B_n(y,t)$, which, to reduce the oscillations, has been smoothed via spline interpolation with gaussian filtering; the right panel shows the numerical measurement for $\dot{L}$ as a function of $q$ compared to the theoretical prediction \eqref{theoretical}. 
The agreement between the analytical predictions and the numerical measurements in the subcritical regime is excellent. The small discrepancies visible in Fig.~\ref{fig:Amplitudes} are to be expected, and arise from the assumptions underlying the analytical theory. 
Specifically, the analytical predictions describe the interaction of ideal line solitons obtained from the exact two-soliton solutions, whereas the numerical simulations involve dispersive shock waves of finite width whose leading edges only asymptotically approach the corresponding soliton profiles. 
Moreover, the interaction takes place over a finite time, so residual modulation effects remain present when the quantities are measured. 
The discrepancy is not uniform over the parameter range. 
As $|q|$ approaches the critical value $q_{\mathrm{cr}}$, the interaction becomes increasingly sensitive to small variations of the incoming soliton parameters because the resonant and ordinary interaction regimes coalesce. 
Consequently, finite-width and finite-time effects become relatively more important, leading to a modest increase in the deviation from the asymptotic predictions. 
Away from the critical regime, the agreement rapidly improves, confirming the robustness of the analytical description.

\subsection{Supercritical type II wedges}
\label{s:SuperII}

\begin{figure}[t!]
    \centering
    \includegraphics[width=\textwidth]{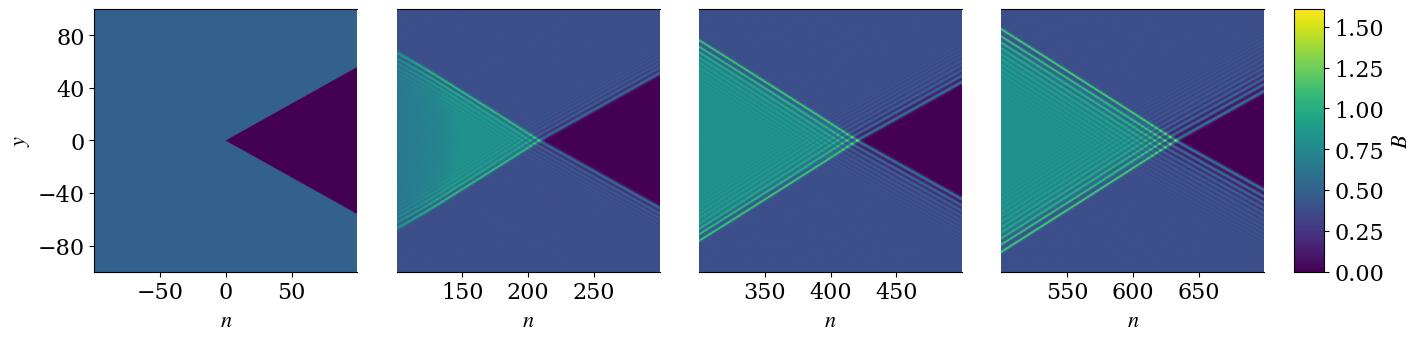}
    \caption{Density plots illustrating the evolution of a supercritical type II wedge with $q = -1.8$. The snapshots are at $t = 0,\ 100,\ 200,\ 300$. Note the absence of a vertical region and the formation of a right propagating localized peak along $y=0$, as well as an expanding multi-phase region surrounding it.}
    \label{fig:SuperII}
\end{figure}

In the supercritical type II regime ($q<0$ and $|q|>q_{\mathrm{cr}}$ corresponding to Case~2 of Remark~\ref{Obs:Cases}) the interaction between the two oblique DSWs is ordinary rather than resonant. No vertical stem DSW is produced; instead, as illustrated in Figure $\ref{fig:SuperII}$, the two DSWs pass through each other, producing a localized but expanding multi-phase region around $y=0$ characterized by high-amplitude peaks that propagate in the positive $n$ direction. This is the discrete analog of the regular reflection described in \cite{BBH_PRL2025}.
In this regime the ordered soliton parameters for the leading-edge solitons in the oblique DSWs are
\begin{equation*}
(P_1,P_2,P_3,P_4)=(-\gamma-\alpha, \alpha-\gamma, \gamma-\alpha, \alpha + \gamma),
\end{equation*}
with $0<\alpha<\gamma$. 
The exact two-soliton solution is given by Eq. \eqref{Bsoliton}
with $N=2$ and $M=4$, which in this case yields
  \begin{equation}
    \label{e:2DTL_ordinary2soliton}
    B_n=\log \left( \frac{\left(\displaystyle\sum_{(i,j)\in S}(e^{P_j}-e^{P_i}) e^{\Theta_{ij}(n+1,y,t)} \right) \left( \displaystyle\sum_{(i,j)\in S}(e^{P_j}-e^{P_i}) e^{\Theta_{ij}(n-1,y,t)}\right)}{\left(\displaystyle\sum_{(i,j)\in S}(e^{P_j}-e^{P_i}) e^{\Theta_{ij}(n,y,t)} \right)^2} \right)
\end{equation}
with $\Theta_{ij}(n,y,t)=\theta_i(n,y,t)+\theta_j(n,y,t)$ and $S=\{(1,3),(1,4),(2,3),(2,4)\}.$ Expression \eqref{Bsoliton} can be written in the compact form (see Appendix \ref{app:C} for details)
\begin{equation} \label{SuperSol}
    B_n(\mathcal U, \mathcal V)=\log \left( \frac{(\sinh\gamma\cosh{(\mathcal U+2\alpha)+\rho\cosh{\mathcal V})}(\sinh\gamma\cosh{(\mathcal U-2\alpha)}+\rho\cosh{\mathcal V})}{(\sinh\gamma\cosh{\mathcal U}+\rho\cosh{\mathcal V})^2} \right)
\end{equation}
where
\bse
\begin{align}
&\mathcal U = 2\alpha\,n+ 2t\,\sinh\alpha\cosh\gamma + \alpha,\qquad
\mathcal V = 2y\,\sinh\gamma\sinh\alpha 
  + \frac12 \log\left[ \frac{\sinh(\gamma+\alpha)}{\sinh(\gamma-\alpha)} \right], 
\\
 &\rho=\sqrt{\sinh^2\gamma-\sinh^2\alpha}\,.
\end{align}
\ese
Note that $\rho$ is real since $|\gamma|>\alpha$ in the supercritical regime.

Using Eq.  \eqref{SuperSol} one can verify by direct computation that the unique stationary point is $U=V=0$ which is the maximum of the solution. 
The condition $U=0$ determines the horizontal velocity of the peak as
\bse
\label{peak}
\begin{align}
    &c_{\mathrm{peak}} = \frac{\sinh\alpha\cosh\gamma}{\alpha}, \\
\noalign{\noindent while evaluation of $B_n$ at the maximum point gives the peak amplitude as}
    &A_{B,\mathrm{peak}} =  
    2\log\left[ 1 + \frac{2\sinh\gamma\sinh^2\alpha}{\rho+\sinh\gamma}\right].
    \end{align}
\ese
As for $A_{B,\mathrm{stem}}$, $A_{B,\mathrm{peak}}$ could equivalently be expressed directly in terms
of $q$ and $\Delta$.
\begin{remark} \label{continuity}
It is interesting to compare the behavior of the peak amplitude $A_{B,\mathrm{peak}}$ from~\eqref{peak} and the stem amplitude $A_{B,\mathrm{stem}}$ from~\eqref{parameters_stem} at the critical slope. 
For the peak amplitude as $\gamma\to\alpha$ 
and for the stem amplitude as $\gamma\to\alpha$ one has
\begin{equation*}
    A_{B,\mathrm{peak}} \big|_{\gamma=\alpha} = 
     A_{B,\mathrm{stem}} \big|_{\gamma=\alpha} = 
      2\log{(\cosh{(2\alpha)})}.    
\end{equation*}
The two expressions coincide, confirming that the maximum amplitude is continuous across the critical threshold. This is consistent with the analogous result for KP equation in \cite{BBH_PRL2025}. 
A similar results applies for the horizontal velocities, for which one has
\begin{equation*}
    c_{\mathrm{peak}} \big|_{\gamma=\alpha} = 
    c_{\mathrm{stem}} \big|_{\gamma=\alpha} = 
    \frac{\sinh{\alpha}\cosh{\alpha}}{\alpha}
\end{equation*}
\end{remark}
The amplitude of the  peak $A_{B,\mathrm{peak}}$ is compared  with the numerical simulations in the right panel of Fig.~\ref{fig:Amplitudes}.
The agreement in this case is not quite as good, similarly to what happens for the KP equation~\cite{BBH_PRL2025}. 
The continuity of $A_{B,\mathrm{stem}}$ and $A_{B,\mathrm{peak}}$ is also visible in the numerical simulations shown in Fig.~\ref{fig:SuperII}. We suspect that a source for the remaining discrepancy is the fact that the analytical prediction isolates the interaction between the two leading-edge solitons, whereas the numerical solution contains the full DSWs, whose oscillatory structure continues to interact with the collision region. These additional finite-amplitude dispersive effects are not captured by the ideal two-soliton description and may therefore produce the small deviations observed. 
A complete understanding of the discrepancies would likely require a genuinely two-dimensional Whitham modulation theory for the 2DTL, capable of describing the modulation of the full dispersive shock wave during the interaction. 
Developing such a theory lies beyond the scope of the present work and is left for future investigation.

\begin{figure}[b!]
    \centering
    \includegraphics[width=\textwidth]{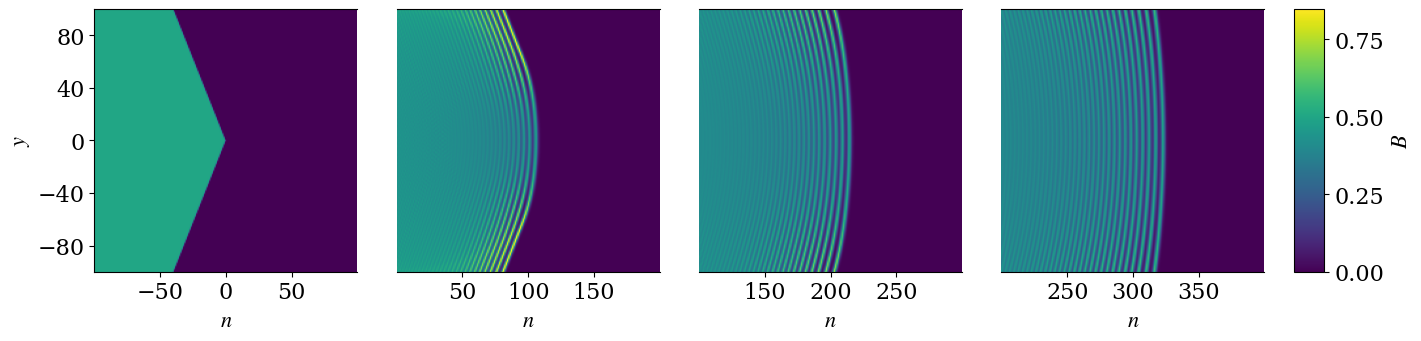}
    \caption{Density plots illustrating the evolution of a subcritical type I wedge with $q = 0.4$. The snapshots are at $t = 0,\ 100,\ 200,\ 300$.}
    \label{fig:SubI}
\bigskip
    \centering
    \includegraphics[width=\textwidth]{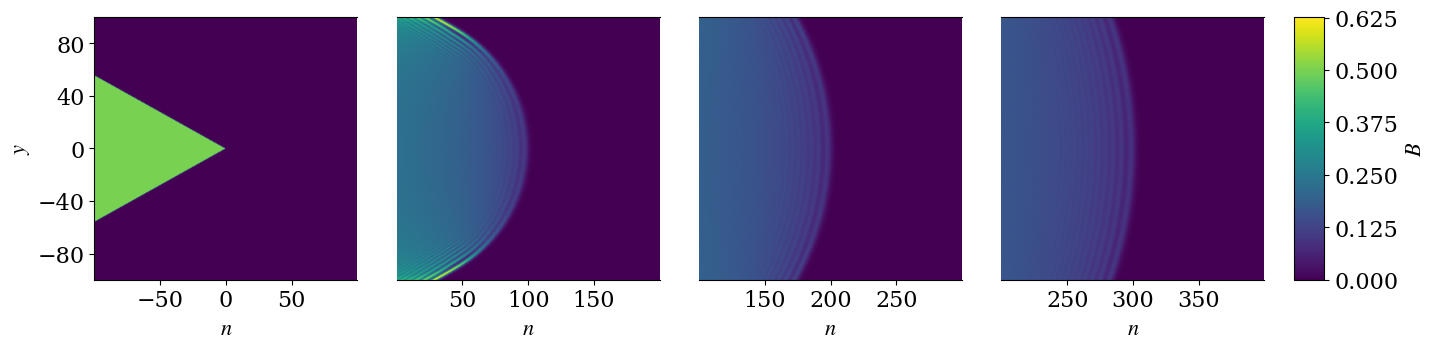}
    \caption{Density plots illustrating the evolution of supercritical type I wedge with $q = 1.8$. The snapshots are at $t = 0,\ 100,\ 200,\ 300$.}
    \label{fig:SuperI}
\end{figure}

\subsection{Type I wedges}\label{TypeI}

Next we briefly discuss the dynamics produced by type I wedges ($q>0$), corresponding to Cases 3 and 4 of Remark \ref{Obs:Cases}. 
In contrast to the type~II wedges, in this case the two oblique DSWs propagate away from the symmetry axis $y=0$, producing a qualitatively less dramatic phenomenology,
as illustrated in Figs.~\ref{fig:SubI} and~\ref{fig:SuperI}.
On the other hand, like with type~II wedges, the qualitative behavior mirrors that observed for the KP equation in \cite{BBH_PRL2025}, with some quantitative  differences due to the discrete nature of 2DTL. 

In subcritical type~I wedges, an example of which is shown in Fig.~\ref{fig:SubI}, 
the time evolution produces a constant-amplitude vertical DSW near $y=0$ which connects smoothly to the two oblique DSWs. As time progresses, the front tends to flatten, approaching a uniform profile in the transverse direction. 
The numerically measured horizontal velocity agrees reasonably well with the theoretical prediction $c_{\mathrm{stem}}$ derived in Section~\ref{s:SubII}, as shown in the left panel of Fig.~\ref{fig:Amplitudes}.

In supercritical type~I wedges, an example of which is shown in Fig.~\ref{fig:SuperI}, the solution near $y=0$ is characterized by a DSW region whose amplitude decays rapidly in time, in analogy with the behavior observed for the KP equation in~\cite{BBH_PRL2025}, where this regime was connected quantitatively to the cylindrical Korteweg-de~Vries (KdV) reduction. 
    
A more precise analytical study and quantitative numerical investigation of both type~I regimes would require the use of Whitham modulation equations for the 2DTL.
Such equations are not available at present, however.
Such a study is therefore left as an open problem for future work.

\section{Continuum limit of the lattice dynamics}
\label{s:continuum}

It is well known that the two-dimensional Toda lattice studied in this work is connected to the KP equation via a long-wave, small-amplitude limit. 
In this section, we review this connection and use this reduction to show that the quantitative analysis provided in Section~\ref{s:wedge} is consistent with what explored in  \cite{BBH_PRL2025}. 
Throughout this section, to clearly distinguish between the parameters of the two models, we adopt the superscripts ``Toda''  and ``\text{KP}'' for quantities referring to the 2DTL and the KP equation respectively. 
Note that in the previous sections these superscripts were omitted since only the Toda lattice was considered.

\subsection{Reduction of the 2DTL to the Kadomtsev-Petviashvili equation}
\label{s:kpreduction}

To derive the continuum limit of the 2DTL, we introduce the multiple scale variables
\begin{equation} \label{multiex}
    X=\epsilon( n - c t), \quad 
    Y= \epsilon^2 y, \quad 
    T = \epsilon^3 t,
\end{equation}
as well as the rescaled amplitude 
\be
\label{e:BexpansionKP}
B_n(y,t) = \epsilon^2 U(X,Y,T),
\ee
where $0<\epsilon\ll 1$ is a small parameter. 
Applying the chain rule and expanding $e^{B_n}$ via Taylor series, the left hand side 
and right hand side of \eqref{eq_Toda1} yield, respectively,
\begin{gather*}
    B_{tt}-B_{yy} = 
        \epsilon^4 c \partial_{XX}U-2\epsilon^6 \partial_X \partial_T U- \epsilon^6 \partial_{YY}U+ o(\epsilon^8),
\\        
    e^{B_{n+1}} - 2 e^{B_n} + e^{B_{n-1}} =
        \epsilon^4U_{XX}+\frac{\epsilon^6}{12}U_{XXXX}+\epsilon^6 (UU_X)_X+o(\epsilon^8).
\end{gather*}
Balancing at order $\epsilon^4$ fixes $c=1$, while at order $\epsilon^6$ one  obtains:
\begin{equation} \label{KPA}
    \left(2U_T+\frac{1}{12}U_{XXX}+UU_X \right)_{X}+ U_{YY}=0
\end{equation}
which is precisely KP equation. 
Introducing the rescaled variables $\Tilde{T}=\frac{1}{2}T$, $\Tilde{X}=\sqrt[3]{12}X$, $\Tilde{Y}=Y$ and $\Tilde{U}=\frac{1}{\sqrt[3]{12}}U$, equation~\eqref{KPA} reduces to the standard KP form:
\begin{equation} \label{KP}
    \left( \Tilde{U}_{\Tilde{T}}+\Tilde{U}_{\Tilde{X}\Tilde{X}\Tilde{X}}+\Tilde{U}_{\Tilde{X}}\Tilde{U}\right)_{\Tilde{X}}+ \Tilde{U}_{\Tilde{Y}\Tilde{Y}}=0
\end{equation}
This connection shows that the 2DTL is a natural discrete counterpart of the KP equation, and provides the framework for comparing the wedge theory developed in Section~\ref{s:wedge}  with the analogous results of \cite{BBH_PRL2025}.

\begin{remark} \label{obs: velocity}
For the purpose of comparing numerical results obtained with the 2DTL to those coming from the KP equation, one also needs to specify an initial velocity consistent with the KP ansatz. 
Differentiating equation~\eqref{e:BexpansionKP} with respect to $t$ yields
\begin{equation*}
  \dot{B}_{n}(y,t) = - \epsilon^3 \partial_X U + \epsilon^5 \partial_T U\,.
\end{equation*}
The term $U_T$ can be eliminated by using the KP Eq. \eqref{KP} to  express $U_T$ in terms of the spatial derivatives.
\end{remark}

\begin{figure}
    \centering
    \includegraphics[width=\linewidth]{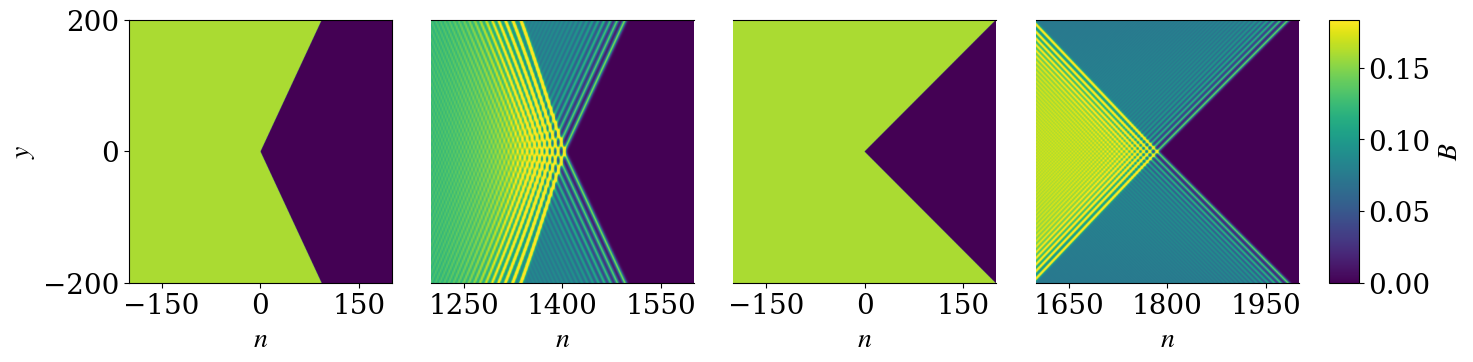}
    \caption{Numerical evolution of the 2DTL  with initial conditions given by the KP  reduction for two representative configurations, both initialized using the KP ansatz \eqref{multiex} with $\epsilon=0.4$ and  $\Delta^\mathrm{Toda}=0.16$. Left two panels: subcritical case $q=-0.45$, showing the initial profile and the solution at time $t=1250$, where the vertical stem has started forming.  Right two panels: supercritical case $q=-1$, showing the initial profile and the solution at time $t=1250$, where instead the peak propagates along $y=0$.}
    \label{fig:wedgeKP}
\end{figure}

\subsection{Consistency of the wedge theory in the Kadomtsev-Petviashvili limit}

The KP reduction established in Section \ref{s:kpreduction} allows to verify the consistency of the critical slope \eqref{Eq:Res_qA}, obtained in Section~\ref{s:wedge}, with the analogous result obtained in \cite{BBH_PRL2025}.  
Recall that in \cite{BBH_PRL2025} it was shown that for the KP equation the critical slope is given by
\begin{equation}\label{qKP}
q^{\mathrm{KP}}_{\mathrm{cr}}=\sqrt{2\Delta^{\mathrm{KP}}},
\end{equation}
where $\Delta^{\mathrm{KP}}$ denotes the size of the initial jump in the Riemann problem for the KP equation. 
Under the multiple scale reduction~\eqref{multiex}, the jump and the slope parameter transform respectively as
\begin{equation}\label{qToda}
    \Delta^{\mathrm{Toda}} = \epsilon^2 \Delta^{\mathrm{KP}},\qquad
    q^{\mathrm{Toda}}=\epsilon q^{\mathrm{KP}},
\end{equation}
where the second relation follows from the fact that $q^{\mathrm{Toda}} = {dn}/{dy} = \epsilon {dX}/{dY}$. 
Substituting Eq. \eqref{qToda} into Eq. \eqref{qKP} yields the small-amplitude limit of the critical angle for the 2DTL as
\begin{equation*}
q_{\mathrm{cr}}^{\mathrm{Toda}}=\sqrt{2\Delta^{\mathrm{Toda}}}.
\end{equation*}
We now verify that equation~\eqref{Eq:Res_qA} reproduces this expression in the small-amplitude limit. Observe indeed that for $\Delta^{\mathrm{Toda}}\ll 1$ one has:
\begin{equation*}
    A^{\mathrm{Toda}}_{B,\mathrm{obl}}=2\log{\left(2e^{\frac{\Delta^{\mathrm{Toda}}}{2}}-1\right)} \approx 2\Delta^{\mathrm{Toda}}
\end{equation*}
and consequently 
\begin{equation*}
   A_V \approx 2\Delta^{\mathrm{Toda}}.
\end{equation*}
Substituting into \eqref{Eq:Res_qA} and using $\arcsinh(x) = x + O(x^3)$ as $x\to0$, one obtains
\begin{align*}
    q^{\mathrm{Toda}}_{\mathrm{cr}} =\frac{A_{V,\mathrm{obl}}}{\arcsinh(\sqrt{A_{V,\mathrm{obl}}})} \approx \sqrt{A_{V,\mathrm{obl}}} \approx \sqrt{2\Delta^{\mathrm{Toda}}}, 
\end{align*}
which is the desired result, confirming the consistency of the critical slope formula with the KP theory of \cite{BBH_PRL2025}. 

\begin{figure}[t!]
\kern-\medskipamount
\centerline{\includegraphics[width=0.9225\linewidth]{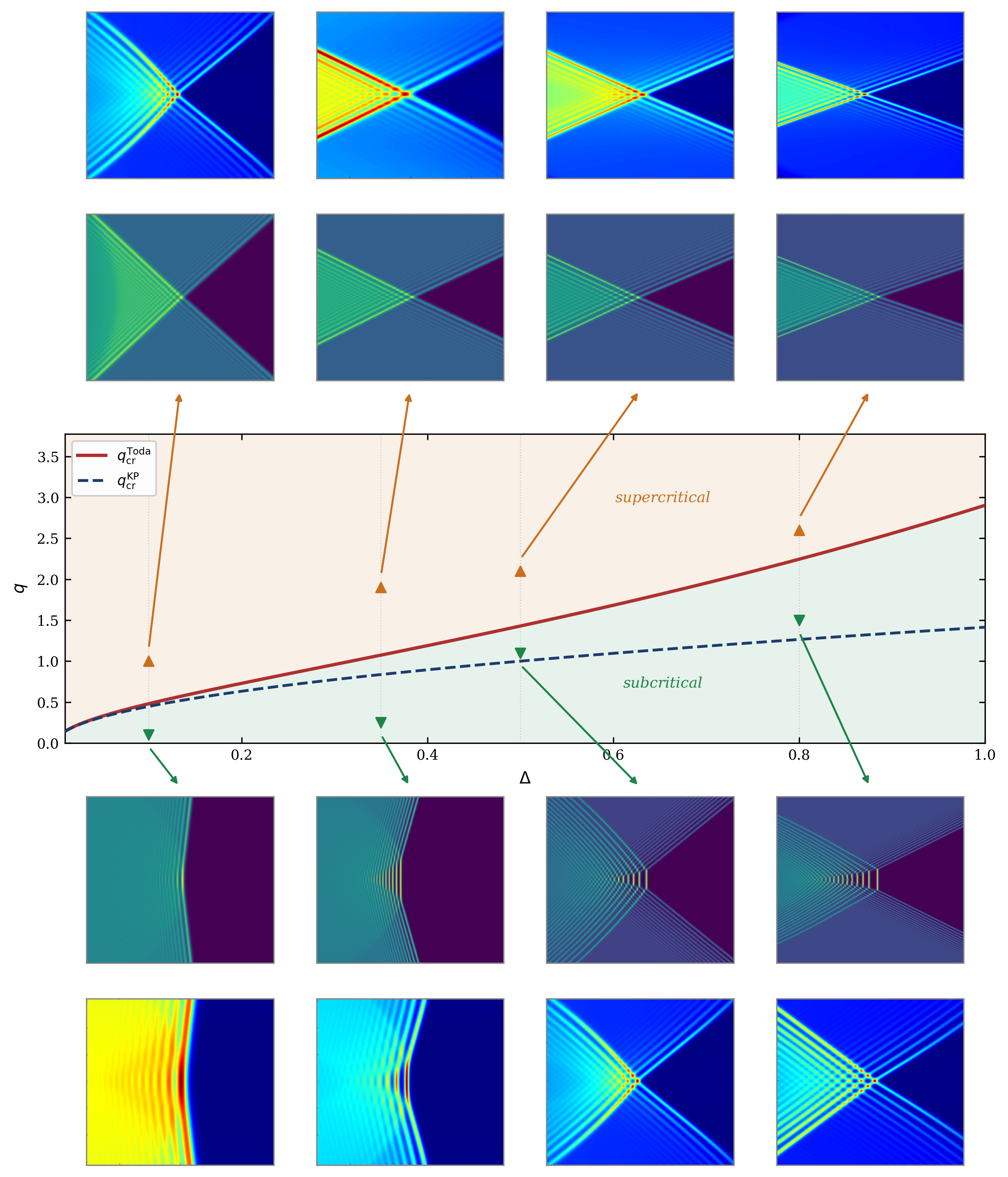}}
\kern2\smallskipamount
\caption{Summary of wedge phenomenology and critical slope diagram. 
    The central panel shows the critical slope $q_{\mathrm{cr}}$ as a function of the jump size $\Delta$ for the 2DTL~\eqref{Eq:Res_qA} (solid red curve) and the KP Eq.~\eqref{qKP} (dashed blue curve). 
    The green and orange shaded regions indicate respectively subcritical and supercritical regimes for the 2DTL. 
    Each arrow points to a corresponding snapshot for the 2DTL (rows closest to the center panel) and the KP equation (rows furthest from the center panel).
    Notably, the right two snapshots in the bottom two rows correspond to parameter choices that are subcritical for the 2DTL but supercritical for the KP equation, highlighting the quantitative discrepancy between the two models. 
    Axes labels are omitted for clarity; the 2DTL snapshots are shown in the $(n,y)$ plane and KP snapshots in the $(X,Y)$ plane, with matching aspect ratios between panels in each column. 
    }
\label{fig:Comparison}
\vglue-4\medskipamount
\end{figure}

We emphasize that, beyond the critical slope, several other qualitative features of the wedge phenomenology are preserved in the KP limit: the resonant versus ordinary nature of the soliton interaction, the continuity of the amplitude across the critical threshold (as pointed out in Remark \ref{continuity}), and the qualitative structure of both type~I and type~II regimes all have direct counterparts in the KP setting. 
On the other hand, the discrete nature of the 2DTL introduces quantitative corrections like in the stem amplitude \eqref{parameters_stem} and in the amplitude formula for the supercritical peak \eqref{peak},  which,
importantly, reduce to the KP expressions in the limit. Two examples are shown in Figure \ref{fig:wedgeKP}, both initialized using the KP ansatz \eqref{multiex} and initial velocity as \eqref{velocity}. In the subcritical case, a vertical stem begins to form, although it remains embryonic at time $t=1250$ since the slow modulation timescale of the KP regime requires significantly longer integration times for the stem to fully develop. In the supercritical case, the localized  peak is produced consistently with the KP theory. Both examples confirm the validity of the multiscale reduction. 

A comprehensive comparison of the two critical slope curves as a function of $\Delta$, together with representative snapshots of the dynamics in both models, is shown in Figure \ref{fig:Comparison},
which compares the critical slope $q_{\mathrm{cr}}$ as a function of the jump size $\Delta$ for the two models. Note the quantitative difference between the two models visible in Fig.~\ref{fig:Comparison}: while the two curves agree for small $\Delta$, consistently with the asymptotic equivalence established in Section \ref{s:kpreduction}, they diverge significantly for larger values of $\Delta$, with the Toda critical slope lying strictly above the KP one. This means that for a given jump size, the subcritical regime is wider for the 2DTL than for the KP equation. Representative snapshots of the dynamics in both models, for parameter choices lying in each of the four regions of the $(\Delta,q)$ plane delimited by the two critical curves, are also shown in Figure \ref{fig:Comparison}. In particular, the rightmost pair of snapshots in the bottom rows corresponds to a parameter choice that is subcritical for the 2DTL but supercritical for the KP equation, lying between the two critical curves: in this case the 2DTL produces a stem DSW while the KP equation produces only a localized peak, providing a direct illustration of the quantitative discrepancy between the two models as the jump
size $\Delta$ increases.

Throughout Sections \ref{s:wedge} and \ref{s:kpreduction} we considered wedges whose two branches have equal and opposite slope $\pm q$. Appendix \ref{app:E} extends the analysis to the more general case in which the two branches have independent slopes $q_1$ and $q_2$ with $q_1q_2<0$, and shows that both the 2DTL and the KP equation possess a symmetry invariance: a Lorentz boost in the $yt$-plane for the former and a pseudo-rotation invariance for the latter. This symmetry maps any such asymmetric wedge onto an equivalent wedge with a single, common slope. As a consequence, the four dynamical regimes identified in Section \ref{s:wedge}, together with their KP counterparts, extend essentially unchanged to the asymmetric case. The two symmetries act differently, however: the KP transformation acts directly and linearly on the physical slopes $q_1$ and $q_2$, so that only the wedge opening angle matters, up to symmetry; the 2DTL Lorentz boost instead acts linearly on the underlying soliton phase parameters, which are related to the physical slopes only through a nonlinear map. Once expressed in the natural variable of each model, the physical slope for the KP equation and the soliton phase parameter for the 2DTL, the two symmetries share exactly the same structure. 

\section{Conclusions and future challenges}
\label{s:conclusions}

In summary, in this work we studied the formation and interaction of DSWs in the 2DTL~\eqref{eq_Toda1} subject to the wedge-type initial conditions~\eqref{InitialValue2dToda}.
We showed that, away from the symmetry axis, the jump across each branch of the wedge acts
locally as a Riemann problem for the one-dimensional Toda lattice, so that the leading-edge
soliton amplitude of each oblique DSW is completely determined by the one-dimensional Whitham modulation theory and depends only on the size of the jump.
The genuinely two-dimensional dynamics emerges from the interaction of the two oblique DSWs
along the axis $y=0$, whose nature is governed by the slope parameter $q$.
We identified two regimes, subcritical and supercritical, and quantified the critical angle $q_\mathrm{cr}$ based on the resonance conditions for the elastic two-soliton solutions of the 2DTL~\cite{BW2010}.
For type~II wedges, in the subcritical regime $|q|<q_{\mathrm{cr}}$ the interaction is resonant, and produces an expanding vertical stem DSW, whose amplitude, vertical expansion rate and horizontal velocity
were computed in closed form, again using the exact two-soliton solutions of the 2DTL.
In the supercritical regime $|q|>q_{\mathrm{cr}}$ the interaction is ordinary, and produces a
localized peak whose amplitude and velocity were obtained analytically.
For type~I wedges the two corresponding regimes were described qualitatively.
Along the way, we also derived a new class of periodic traveling wave solutions of the 2DTL
that generalizes the genus-one solutions of the one-dimensional lattice.
The analytical predictions were confirmed by direct numerical simulations, and the continuum
limit of the results was shown to be consistent with the corresponding theory for the
KP equation, thereby providing an independent validation of the analytical framework. Nevertheless, the quantitative differences between the
2DTL and the KP settings were also illustrated as the size of
the initial jump $\Delta$ was increased.

We believe that this work opens up a number of interesting follow-up problems.   For example, these include the following:

\textit{Transverse stability of the oblique periodic traveling waves.}
The oblique periodic traveling wave solutions of the 2DTL derived in
Section~\ref{s:2DTL_background} are the natural two-dimensional analogues of the genus-one
solutions of the one-dimensional lattice, and their transverse (in)stability is a prerequisite
for understanding the robustness of the DSW structures studied here. It would be interesting to
determine whether, and in which parameter ranges, these solutions are stable to long-wavelength
transverse perturbations.

\textit{Whitham modulation equations for the 2DTL and its applications.}
The periodic traveling wave solutions derived in this work provide the natural starting point
for the derivation of the Whitham modulation equations of the 2DTL, in the same spirit as the
two-dimensional modulation systems recently obtained for several continuous nonlinear evolution
equations~\cite{JPA2023v56p025701,JPA2018v51p215501,RSPA2017v473p20160695,PRE2017v96p032225,MJA-Cole-Rumanov,SAPM2024v152p597}.
Such a discrete two-dimensional Whitham system would put the soliton-based
description of the wedge dynamics on a firmer modulation-theoretic footing, and in particular
would furnish a fully quantitative description of both type~I regimes, which the present analysis only captured qualitatively.
Once available, the resulting modulation equations, or suitable reductions thereof, could be
used to study a variety of dynamical problems, as has been done in the continuous
setting~\cite{RSPA2022v478v20210823,NLTY2021v34p3583,JFM2021v909pA24}.
Particularly interesting would be the study of supercritical type~I regime, for which the oscillatory region near $y=0$ could be governed by a cylindrical KdV-type reduction analogous to the one identified for the KP equation in~\cite{BBH_PRL2025}.
The availability of the Whitham modulation equations for the 2DTL would also provide a quantitative framework to study the stability of the oblique periodic traveling wave solutions of the 2DTL,
similarly to what was recently done for the KP and other (2+1)-dimensional evolution equations in
\cite{RSPA2017v473p20160695,PRE2017v96p032225,SAPM2024v152p597,MJA-Cole-Rumanov,SAPM2024v152p597}.

Another intriguing further direction concerns the study of the elliptic two-dimensional Toda lattice. 
While in the hyperbolic case of the present work resonant interactions can generate stem DSWs, the elliptic case may support qualitatively different dynamics. 
In analogy with the transition from KPII to KPI, where line-soliton interactions are closely connected to the emergence of lump solutions, it would therefore be interesting to investigate whether wedge-type initial data in the elliptic 2DTL can generate localized coherent structures through similar resonant mechanisms. 
Such structures would provide a natural discrete analogue of the KPI lumps \cite{Pang_2023, https://doi.org/10.1002/mma.6370}.

More generally, the fact that phenomena such as stem DSWs appear in both the KP equation and 2DTL suggests that these might be rather general features of (2+1)-dimensional integrable systems,
like resonant multi-soliton solutions and web structure, which are found not only in the KP equation \cite{JPA36p10519,Kodama2017KPSA}
and in the 2DTL \cite{BW2010}, but also in the Davey-Stewartson system \cite{JPA2022v55p305701} and others 
\cite{IsojimaWilloxSatsuma,KodamaMaruno}.
It will therefore be an interesting question whether stem DSWs are also found in other (2+1)-dimensional integrable systems. 

Finally, we should emphasize that, since integrable systems are amenable to exact analytical treatment, they often provide a useful paradigm for the study of more general physical systems. 
We therefore expect that, similarly to the continuous setting, the phenomenology described here will persist, at least qualitatively, in non-integrable two-dimensional lattices, where the local Riemann structure that underlies the oblique DSWs is generic, even if the exact soliton machinery used to
characterize their interaction is not available.
Indeed, it would be interesting to investigate ---potentially via
numerical computations--- whether the same dichotomy between resonant
stem formation and ordinary interaction arises in non-integrable ones such as two-dimensional Fermi--Pasta--Ulam--Tsingou-type chains, for which the local one-dimensional Whitham theory is available but the exact
multi-soliton solutions are not. An  initial attempt along
this direction can be found, e.g., in the very recent work
of~\cite{yang2026localizedpatternsdispersivestructures}.
A comparison across lattice geometries would help disentangle which features of the wedge phenomenology are genuinely integrable and which are robust consequences of discreteness and dimensionality.

We hope that the results of this work and the above discussion will inspire further research along these directions.

\section*{Acknowledgments}
We thank Alexander Bivolcic for interesting discussions about the numerical methods and for providing the figures on the results of KP simulations that were reported in Figure~\ref{fig:Comparison}.
 This work was partially supported by the Simons Foundation under grant numbers SFI-MPS-TSM-00013369 (G.B.) and SFI-MPS-SFM-00011048 (P.G.K.).
This material is also based upon work supported by the US National Science Foundation under grant numbers DMS-2107945 (C.C.) and  PHY-2408988 (P.G.K.). 
This research was partly conducted while P.G.K.\ was  visiting the Okinawa Institute of Science and
Technology (OIST) through the Theoretical Sciences Visiting Program (TSVP), the University of
Sydney through the visitor program of the Sydney Mathematical Research Institute (SMRI) and the Department of Mechanical Engineering at Seoul National University through a Fulbright Fellowship. Their support is gratefully acknowledged.

\addcontentsline{toc}{section}{Appendix}
\section*{Appendix}
\setcounter{section}1
\setcounter{subsection}0
\setcounter{equation}0
\def\theequation{\Alph{section}.\arabic{equation}}
\def\thesection{\Alph{section}}

\subsection{Numerical algorithms}
\label{app:D}

Here we describe the details of the numerical scheme used to simulate the two-dimensional Toda lattice equation in the formulation of equation~\eqref{eq_Toda1}. 
The computational domain consists of $J=2400$ lattice sites in the $n$ direction and $H=6000$ uniformly spaced points in $y$ with spacing $\Delta y=0.25$, chosen large enough to contain all the dynamically relevant structure within an integration window $t\in [0,300]$.

\paragraph{Initial conditions.}
The initial field taken to be a smooth version of the wedge profile \eqref{InitialValue2dToda}. 
Introducing the variable $s=n+q|y|$, the initial numerical profile is:
\begin{equation*} \label{Initial_Num_profile}
    B_n(y,0)=B_-+\tfrac12(B_+-B_-)\big[ 1-\tanh(\delta s) \big],
\end{equation*}
where $\delta$ is the transition steepness parameter. 
The smoothing of the step via a hyperbolic function is not merely a numerical convenience: 
it eliminates the spurious small-scale oscillations that a sharp jump would otherwise generate and that would persist throughout the evolution. 
Crucially, this smoothing does not affect the analytical predictions, since the Whitham modulation theory depends only on the asymptotic states $B_\pm$ and not on the fine structure of the transition layer.
We used the value $\delta = 3$ in all the numerical simulations.

Another aspect that must be addressed is the fact that \eqref{eq_Toda1} is second-order in time, specifying $B_n(y,0)$ alone does not determine the solution: one must also prescribe $\dot{B}_n(y,0)$. To isolate the right-going dynamics, we exploit the Flaschka variable formulation of the 1DTL \eqref{Toda1dFlaschka}: using the relation $\dot{B}_n=a_{n+1}-a_n$ together with Eq. \ref{velocity},  we fix the Riemann invariant $\lambda_1$ to be constant across the jump, as described in Section 4.

\paragraph{Time integration.}
Since \eqref{eq_Toda1} is a second-order hyperbolic system possessing a Hamiltonian structure, we integrate it in time using the Verlet scheme. 
Denoting $B_n^j=B_n(\cdot,j\Delta t)$, 
given the field $B_n$ at two levels $B_n^j$ and $B_n^{j-1}$,
and taking the first update as 
$B_n^1=B_n^0+\dot{B}_n^0\Delta t+ \tfrac12{\Delta t^2}\mathcal{L}(B_n^0)$, 
the $j$-th update is
\begin{gather}
    B_n^{j+1}=2B_n^j-B_n^{j-1}+ \Delta t^2 \mathcal{L}(B_n^j),
\end{gather}
where $\mathcal{L}(B_n^j):=\partial_{yy}B_n^j+D^2_n[e^{B_n^j}]$,
the operator $D^2_n$ denotes the standard second-order discrete Laplacian (i.e., $D_n^2[f]=f_{n+1}-2f_n+f_{n-1}$) 
and $\Delta t$ is the time step. 
The scheme is symplectic: it preserves the discrete counterpart of the symplectic structure (i.e., the symplectic two-form) of the Toda lattice. The scheme is also explicit and second-order accurate in time.
In all simulations, we used the value $\Delta t= 0.25$. The spatial discretization in the transverse direction $y$ is described below, together with the absorbing boundary conditions.

\paragraph{Boundary conditions in the lattice direction.}
At the two endpoints $n=0$ and $n=J-1$ we impose homogeneous Neumann conditions on $e^{B_n}$ implemented via ghost points $e^{B_{-1}}:=e^{B_1}$ and $e^{B_{J}}:=e^{B_{J-2}}$, giving%
\bse
\begin{equations}
        D_0^2\left (e^{B_0}\right )=2\left (e^{B_1}-e^{B_0} \right ),\\
        D_{J-1}^2\left (e^{B_{J-1}} \right )=2\left (e^{B_{J-2}}-e^{B_{J-1}}\right ).
\end{equations}
\ese
Since the wedge initial data are designed to suppress the left-going DSW and isolate the right-going branch, the resulting stem (respectively, peak) propagates toward increasing values of $n$. In our simulations, the computational domain is chosen sufficiently large so that no boundary-generated disturbance reaches the interaction region over the time interval $0\le t\le 300$.

\paragraph{Absorbing boundary conditions in the transverse direction.}
The transverse direction requires a different treatment, since the wedge initial data radiates energy in the $y$ direction and periodic or reflective boundary conditions would contaminate the long-time dynamics with spurious reflections. We therefore surround the domain with a Convolutional Perfectly Matched Layer (CPML)  \cite{CPML}, which absorbs outgoing waves without reflection at the physical layer interface. Inside the absorbing layer, the $y-$derivative is replaced by a complex-stretched version:
\begin{equation}
    \partial_y \to \frac{1}{\Tilde{\kappa}(y)}\partial_y, \qquad 
    \tilde{\kappa}(y)=\kappa(y)+\frac{\sigma(y)}{\alpha(y)+i\omega},
\end{equation}
where $\sigma(y) \ge 0$ provides the exponential damping of the propagation waves. It vanishes on the physical boundary and grows inside the layer. The frequency shift parameter $\alpha(y)\ge 0$ regularizes the pole at $\omega=0$, and $\kappa(y)\ge 1$ is the stretching factor. In our implementation the CPML occupies $14\%$ of the transverse computational domain at each boundary. Within each absorbing layer, we introduce a normalized coordinate $\upsilon \in [0,1]$, measured from the interface with the physical domain ($\upsilon=0$) toward the outer boundary ($\upsilon=1$). The stretching factor $\kappa$ follows a cubic polynomial profile:
\begin{equation}
    \kappa(\upsilon)=1+ 4\upsilon^3
\end{equation}
so that $\kappa$ increases from $1$ at the physical domain (CPML interface) to $5$ at the outer boundary. The damping and the frequency-shift parameter $\sigma$ and $\alpha$ are also graded using cubic profiles. The resulting transformation is implemented via an auxiliary variable $\phi_n(y,t)$ that satisfies the update rule
\begin{equation} \label{Recursion}
    \phi_{n+1}=e^{-(\frac{\sigma}{\kappa}+\alpha)\Delta t}\phi_n+\frac{\sigma \left(e^{-(\frac{\sigma}{\kappa}+\alpha)\Delta t}-1\right)}{\kappa(\sigma +\kappa \alpha)}(\partial_y B)_n\,.
\end{equation}
More details on the CPML method and its implementation can be found in \cite{CPML}. The transverse coordinate  $y$ is discretized on a uniform grid of $H$ points with spacing $\Delta y$. The first derivative $\partial_y  B$ is computed via second-order central differences, the auxiliary memory variable $\phi_n$ is updated via the formula \eqref{Recursion}. A second order central difference operator is then applied to the flux $\tfrac{\partial_y B}{\kappa}+\phi_n$ to obtain the discrete approximation of $\partial_{yy}B$

\subsection{Further background on the one-dimensional Toda lattice}
\label{app:A}

\paragraph{Genus-two regularization.}
%
For completeness, here we report the genus-two regularization of the Riemann problem \eqref{RiemannProblemLambda}, which arises when the two intervals 
\begin{equation*}
L_1:=[\min{(\lambda_1^-,\lambda_1^+)},\max{(\lambda_1^-,\lambda_1^+)}],\qquad
L_2:=[\min{(\lambda_2^-,\lambda_2^+)},\max{(\lambda_2^-,\lambda_2^+)}]
\end{equation*}
are disjoint and $\lambda_1^- \le \lambda_1^+$ or $\lambda_2^- \le \lambda_2^+$.  
In this case, a full regularization of the problem requires the embedding of the initial data \eqref{RiemannProblemLambda} into a degenerate Whitham genus-two system. 
Denoting
$(\lambda_1^o,\lambda_2^o,\lambda_3^o,\lambda_4^o):=\mathrm{ord}(\lambda_1^-,\lambda_2^-,\lambda_1^+,\lambda_2^+)$ where $\mathrm{ord} (v_1,...,v_n):=(v_{i_1},....,v_{i_n})$ with $v_{i_1}\le v_{i_2} \le ... \le v_{i_n}$, the solution is obtained,
similarly to \cite{AMBYK},
by considering the genus-two Whitham system with non-increasing initial conditions:
\begin{align*} 
    &\lambda_1(x,0)=\lambda_1^o,\qquad  
    \lambda_2(x,0)=\begin{cases}
        \lambda_2^o, &x\le0, \\ \lambda_1^o, &x>0,
    \end{cases}\qquad
    \lambda_3(x,0)=\lambda_2^o,
    \\ 
    & \lambda_4(x,0)=\lambda_3^o,\qquad
    \lambda_5(x,0)= \begin{cases}
        \lambda_4^o, &x\le0, \\ \lambda_3^o, &x>0,
    \end{cases}\qquad
    \lambda_6(x,0)=\lambda_4^o\,.
    \end{align*}
Since these initial conditions are non-increasing, the Whitham system \eqref{WhitamTodaGenusg}  admits a global solution in $\mathbb{R}\times [0,\infty)$ given by
\begin{align*} 
    &\lambda_1(x,t)=\lambda_1^o,\qquad
    \lambda_2(x,t)=\begin{cases}
        \lambda_2^o \ \ x\le s_2^-t, \\ \lambda_1^o \ \ x>s_2^+t,
    \end{cases}\qquad
    \lambda_3(x,t)=\lambda_2^o,\\
    & \lambda_4(x,t)=\lambda_3^o,\qquad 
    \lambda_5(x,t)= \begin{cases}
        \lambda_4^o\ \ x\le s_5^-t, \\ \lambda_3^o \ \ x>s_5^+t,
    \end{cases}\qquad 
    \lambda_6(x,t)=\lambda_4^o,
    \end{align*}
where $\lambda_1(x,t), \lambda_3(x,t), \lambda_4(x,t)$ a and $\lambda_6(x,t)$ remain constant while $\lambda_2$ and $\lambda_5$ vary smoothly in the expanding regions $s_2^-t < x < s_2^+t$ and $s_5^- <x < s_5^+ t$ respectively, these regions can be obtained using a similar computation as in Section~\ref{s:Riemann1DTL}.

\paragraph{Amplitudes of the leading edge solitons of the DSW in Flaschka variables}
%
Here we prove the formulae for the extremal values of the Flaschka variables for each DSW produced by the ICs~\eqref{RiemannProblemFlaschka} in Section \ref{s:Riemann1DTL}.
These expressions are then used in Section~\ref{s:wedge}.

\begin{proposition}
The Riemann problem \eqref{RiemannProblemFlaschka} produces two counter-propagating DSWs.
The extremal values of the Flaschka variables within each DSW are:
\bse
\begin{align}
    &a^L_{\mathrm{max}} = -2\sqrt{b_-}+a_++2\sqrt{b_+},\qquad
    a^L_{\mathrm{min}} = a_-,\qquad
    b_{max}^L = \tfrac14{(a_+-a_-+2\sqrt{b_+})^2},\qquad
    b_{\mathrm{min}}^L = b_- \\
    &a^R_{\mathrm{max}} = -2\sqrt{b_+}-a_-+2\sqrt{b_-},\qquad
    a^R_{\mathrm{min}} = -a_+,\qquad
    b^R_{\mathrm{max}} = \tfrac14{(a_+-a_-+2\sqrt{b_-})^2},\qquad
    b_{\mathrm{min}}^R = b_+.
\end{align}
\ese
\end{proposition}

\begin{proof}
The extremal values of $a_n$ follow from equation~\eqref{e:a_n(t)_periodic},
which, combined with the bounds $0\le sn(Z_n(t),\gamma)\le 1$, give
\begin{align*}
    a_{\mathrm{min}}=\tfrac{1}{2}(\lambda_1+\lambda_2-\lambda_3+\lambda_4),\qquad
    a_{\mathrm{max}}=\tfrac{1}{2}(\lambda_1-\lambda_2+\lambda_3+\lambda_4) 
\end{align*}
Substituting the appropriate Riemann invariants in the solitonic limit yields the expression for the left extremal values.

For the $b_n$ coordinate, in the solitonic limit of the periodic solutions \eqref{PeriodicSolutions} one has: 
\begin{equation*}
    b(\mu) = 
    \tfrac{1}{4} \left( \tfrac12\lambda^2_1+\tfrac12\lambda^2_2+\lambda^2_4 \right)
    + 2R(\mu)-\mu^2-\tfrac{1}{4}a^2(\mu),
\end{equation*}
where $R(\mu)=(\tfrac12\lambda_4-\mu)\sqrt{(\mu-\frac12\lambda_1)(\mu-\frac12\lambda_2)}$ and $a(\mu)=\tfrac12\lambda_1+\tfrac12\lambda_2+\lambda_4-2\mu$, using the fact that, in the degenerate limit, of the genus-one Whitham equations, the left Riemann invariants satisfy $\lambda_1 < \lambda_2 < \lambda_3 \equiv\lambda_4$.
The minimum value of $b$ is attained  at $\mu =\frac12\lambda_4$, giving
\begin{equation*}
    b^L_{\mathrm{min}} =
    \tfrac{1}{4} \left( \tfrac12\lambda^2_1+\tfrac12\lambda^2_2+\lambda^2_4 \right)
    - \tfrac14\lambda^2_4 - \tfrac1{16}(\lambda_1+ \lambda_2)^2 = \tfrac1{16}(\lambda_2-\lambda_1)^2.
\end{equation*}
Upon substituting $\lambda_1=\lambda_1^-=a_--2\sqrt{b_-}$ and $  \lambda_2=\lambda_2^-=a_-+2\sqrt{b_-}$  one obtains $b^L_{\mathrm{min}}=b_-$, as claimed.
For the maximum value $b_{\mathrm{max}}^L$ one has:
\begin{equation*}
    \frac{db}{d\mu}=2 R'(\mu)-2\mu+a(\mu)
\end{equation*}
with $R'(\mu)= -\sqrt{(\mu-\frac12\lambda_1)(\mu-\frac12\lambda_2)}
  +  \frac14\left( \frac12\lambda_4-\mu \right)\big[4\mu - \lambda_1 - \lambda_2\big]\big/\sqrt{(\mu-\frac12\lambda_1)(\mu-\frac12\lambda_2)}.$
Setting ${db}/{d\mu}=0$ yields the critical point
\begin{equation*}
\mu_*=\frac{\lambda_4^2-\lambda_1\lambda_2}{2(2\lambda_4-\lambda_1-\lambda_2)},
\end{equation*}
at which
\begin{equation*}
    b^L_{\mathrm{max}}:=b(\mu_* )=\tfrac1{16}{(2\lambda_4-\lambda_1-\lambda_2)^2}.
\end{equation*}
Upon substituting the expressions for $\lambda_1,\lambda_2,\lambda_4$ written in terms of $a_\pm,b_\pm$ one obtains $ b^L_{\mathrm{max}}=\frac14{(a_2-a_1+2\sqrt{b_2})^2}$ as claimed. 

The values for the right propagating DSW are obtained via the symmetry $\Tilde{a}_n(t)=-a_{-n}(t)$ and $\Tilde{b}_{n}(t)=b_{-n+1}(t)$, which corresponds to replacing $(a_-,b_-,a_+,b_+)$ with $(-a_+,b_+,-a_-,b_-)$
\end{proof}

\subsection{Soliton solutions of the two-dimensional Toda lattice in laboratory coordinates}
\label{app:B}

\paragraph{One-soliton solution in laboratory coordinates.}
We derive the one-soliton solution of the 2DTL in laboratory frame. The line soliton solutions for \eqref{e:Toda2D_0} in light cone variables are given by \cite{BW2010}:
\begin{equation} \label{1Soliton}
    V_n(r,s)=\frac{1}{4}a \omega \sech^2\left [\frac{1}{2}a \left(r+dn+\frac{\omega  s}{a}-r_0 \right) \right ]   
\end{equation}
where $a,\omega,d$ and $r_0$ are real parameters. Passing to laboratory frame coordinates via $r=\frac{t+y}{2}$ and $r=\frac{t-y}{2}$, Eq. \eqref{e:Toda2D_0} becomes:
\begin{equation}\label{Toda3yt}
    \Box\log{(1+V_n)}=V_{n+1}-2V_n+V_{n+1}
\end{equation}
where $\Box:= \partial^2_t-\partial_y^2$ is the wave operator. Applying the same transformation of variables to the soliton solution \eqref{1Soliton} yields:
\begin{equation}\label{1Solitonyt}
    \Tilde{V}_n(y,t) =A\sech^2\left [\Tilde{a}\left(\left(\frac{1}{2} -\frac{\omega}{4\Tilde{a}}\right)y+dn+\left(\frac{1}{2}+\frac{\omega
    }{4\Tilde{a}} \right)t-\frac{t_0+y_0}{2} \right) \right ] 
\end{equation}
with $A:= \frac14a\omega$, $\tilde{a}=\frac12{a}$ and where 
$\tilde V_n(y,t):=V_n\big(\frac12(t+y),\frac12(t-y)\big)$.

Following \cite{BW2010}, the parameters $a$ and $\omega$ are expressed in terms of two positive real parameters $p_1<p_2$:  
\begin{equation*}
a = p_2-p_1, \qquad 
\omega = \frac{1}{p_1}-\frac{1}{p_2}.
\end{equation*}
Using this parameterization, the solution \eqref{1Solitonyt} becomes
\begin{align*}
\tilde V_n(y,t) = A(p_1,p_2)\sech^2 \left [ \Tilde{a}(p_1,p_2) \left( \left( \frac{1}{2}-\frac{1}{2p_1p_2}\right)y+dn+ \left( \frac{1}{2}+\frac{1}{2p_1p_2}\right)t - \frac{t_0+y_0}{2} \right)\right ],
\end{align*}
with 
\begin{equation*}
A(p_1,p_2) = -\frac{1}{4}(p_2-p_1)\left( \frac{1}{p_1}-\frac{1}{p_2}\right),\qquad
\Tilde{a}(p_1,p_2)=\frac{p_2-p_1}{2}.
\end{equation*}
This is the parameterization adopted in \cite{BW2010}.
For our purposes, it is convenient to introduce a logarithmic parameterization given by $P_i=\log{p_i}$ for $i=1,2$. Then, using the identities
\begin{gather*}
    \frac12 \left( \frac{1}{p_i}+p_i \right) = \frac{e^{P_i}+e^{-P_i}}{2}=\cosh{P_i}\,,\qquad
    \frac12 \left( \frac{1}{p_i}-p_i \right) = \frac{-e^{P_i}+e^{-P_i}}{2}=-\sinh{P_i}\,,
\\
A = \frac{1}{4}(p_i-p_j)\left( \frac{1}{p_j}-\frac{1}{p_i}\right)
    =\left( \frac{p_i-p_j}{2\sqrt{p_ip_j}}\right)^2 
    = \frac14 \left( e^{\frac12(\log{p_i}-\log{p_j})}-e^{-\frac12(\log{p_i}-\log{p_j})} \right)^2 
    =  \sinh^2\left( \frac{P_i-P_j}{2} \right),
\end{gather*}
one immediately obtains equation~\eqref{SolitonHyperbolic}.

This waveform constitutes the two-dimensional Toda lattice parameterized by $P_1$ and $P_2$. As we showed in Section~\ref{s:wedge}, these solutions combined with the one-dimensional Whitham modulation theory presented in Section~\ref{s:1DTL} allowed us to analytically obtain the leading-edge line solitons formed during the evolution of the wedge problem explored in Section~\ref{s:wedge}.

\paragraph{Inversion formulae.}
Here we derive the inversion formulae~\eqref{InverseRelation} that allow one to express the soliton parameters $P_1$ and $P_2$ in terms of the physically significant quantities amplitude $A$ and slope $q$ defined in \eqref{Amplitude} and \eqref{Slope}.

\begin{proposition}
The system
\begin{equation}\label{physicalquant.}
    A=\sinh^2 \left( \tfrac12(P_2-P_1) \right)\,,\qquad 
    q=\frac{\cosh{P_2}-\cosh{P_1}}{P_2-P_1}\,,
\end{equation}
can be inverted for $P_2>P_1$ to give \eqref{InverseRelation} with $\alpha$ and $\gamma$ as in \eqref{e:alphagammafromAq}.
\end{proposition}

\begin{proof}
Define the following quantities:
\begin{equation*}
\nu:=P_2-P_1\,,\qquad 
\gamma:=\tfrac12(P_1+P_2)\end{equation*}
From the first equation, one has $\sinh{\left(\frac{\nu}{2} \right)}=\sqrt{A}$, hence $\nu=2\alpha$ since we have $P_2>P_1$ by hypotheses. 
Using the identity 
$\cosh{u}-\cosh{v} = 2 \sinh\left( \frac12(u+v) \right)\sinh{\left( \frac12(u-v) \right)}$, the slope formula becomes
\begin{equation*}
    q=\frac{\cosh{P_2}-\cosh{P_1}}{P_2-P_1}=\frac{2\sinh\gamma\sinh{\left(\frac{\nu}{2} \right)}}{\sigma2\alpha}= \frac{\sqrt{A}}{\alpha} \sinh\gamma
\end{equation*}
which gives
\begin{equation*}
\gamma=\arcsinh\left( q \frac{\arcsinh(\sqrt{A})}{\sqrt{A}}\right)\end{equation*}
Inverting the definitions for $\Delta(P_1,P_2)$ and $\gamma(P_1,P_2)$ and using the last expression for $\gamma$ one has
\begin{gather*}
    P_1=\arcsinh \left( q\frac{\arcsinh(\sqrt{A})}{\sqrt{A}}\right)- \arcsinh(\sqrt{A}), \qquad
    P_2= \arcsinh \left( q\frac{\arcsinh(\sqrt{A})}{\sqrt{A}}\right)+ \arcsinh(\sqrt{A}),
\end{gather*}
which is what we claimed, namely \eqref{InverseRelation}.
\end{proof}

\begin{remark}  
The choice $P_2>P_1$ is made without loss of generality, since \eqref{SolitonHyperbolic} is symmetric under the exchange $P_1 \leftrightarrow P_2$.   
\end{remark}

\subsection{Quantitative analysis of subcritical and supercritical type~II wedges}
\label{app:C}

\paragraph{Subcritical type II wedges.}

Here we derive the analytical expressions for the time  evolution of the upper and  lower edges of the stem DSW in  the subcritical type II regime, which are  used in Section~\ref{s:SubII}. Recall that in this regime the soliton parameters are  $(P_1,P_2,P_3,P_4)=(\gamma-\alpha,-\gamma-\alpha,\alpha + \gamma,\alpha - \gamma)$ with $-\alpha<\gamma<0$, and the phases of the two-soliton solution are:
\begin{equation*}
\theta_i=nP_i+y\cosh{P_i}-t\sinh{P_i}+\theta_{0i} \ \ \ i=1,2,3,4
\end{equation*}
The bottom edge $(n_-(t),y_-(t))$ is the triple point at which the three solitons with  index  pairs $[1,2],[1,4]$ and $[2,4]$ meet, i.e., the point where  $\theta_1=\theta_2=\theta_4$. This yields the system
\begin{equation*}
    \theta_1=\theta_4,\qquad 
    \theta_1=\theta_2.
\end{equation*}
The first condition above yields 
\begin{equation} \label{theta1theta4}
    n(P_1-P_4)+y(\cosh{P_1}-\cosh{P_4)}-t(\sinh{P_1}-\sinh{P_4})+\theta_{01}-\theta_{04}=0
\end{equation}
Now observe the following:
\begin{equation*}
    P_4-P_1=2(\gamma+\alpha), \qquad 
    \cosh{P_4}-\cosh{P_1}=0, \qquad
    \sinh{P_4}-\sinh{P_1}=2\sinh{(\gamma+\alpha)},
\end{equation*}
which, substituting in Eq. \eqref{theta1theta4} gives:
\begin{equation*}
    2(\gamma+\alpha)n + 2\sinh{(\gamma+\alpha)}t+\theta_{4,0}-\theta_{1,0}=0,
\end{equation*}
from which one obtains
\begin{equation}\label{nCoordinate}
 n_-(t)=\frac{\sinh{(\gamma+\alpha)}}{\gamma+\alpha}t+n_-^0   \qquad 
 n_-^0 = \tfrac12(\theta_{4,0}-\theta_{1,0}),
\end{equation}
which is the equation of the evolution of the $n-$coordinate dynamics for the top edge of the stem.

Next, considering the equation $\theta_1=\theta_2$  one has
\begin{equation}\label{theta1theta2}
    n(P_1-P_2)+y(\cosh{P_1}-\cosh{P_2)}-t(\sinh{P_1}-\sinh{P_2})+\theta_{01}-\theta_{02}=0
\end{equation}
from which we compute 
\begin{equation*}
    P_2-P_1 = 2\gamma, \qquad
    \cosh{P_2}-\cosh{P_1} = 2\sinh\gamma \,\sinh\alpha,\qquad
    \sinh{P_2}-\sinh{P_1} = 2\cosh\alpha \,\sinh\gamma,
\end{equation*}
which, substituting in Eq. \eqref{theta1theta2} gives:
\begin{equation*}
    2\gamma n-2y\sinh\gamma\sinh\alpha-2t\sinh\gamma\cosh\alpha+\theta_{2,0}-\theta_{1,0}=0
\end{equation*}
Substituting the expression for $n_-(t)$ given from Eq. \eqref{nCoordinate} one obtains 
\begin{equation} \label{yCoordinate}
 y_-(t) = - \frac{\alpha\coth{\alpha-\gamma\coth{\gamma}}}{\gamma+\alpha} t+y_-^0 \qquad 
 y_-^0 = \frac{(\theta_{4,0}-\theta_{1,0})\alpha+\theta_{2,0}-\theta_{1,0}}{2\sinh\gamma\sinh\alpha},
\end{equation}
which is the equation of the evolution of the $y-$coordinate dynamics for the lower edge of the stem.

Equations \eqref{nCoordinate} and \eqref{yCoordinate} can  be obtained for the upper edge of the stem by solving the system for the triple point  $[2,3],[2,4]$ and $[3,4]$.

\paragraph{Supercritical type II wedges.}
%
Here we derive the analytical expressions for the peak amplitude and horizontal velocity in the supercritical type II regime, stated in Section \ref{s:SuperII} as Eq. \eqref{peak}. In this regime, the soliton parameters are $(P_1,P_2,P_3,P_4)=(-\gamma-\alpha,\alpha - \gamma,\gamma-\alpha,\alpha + \gamma)$ with $0<\alpha<\gamma$. 
The exact ordinary two-soliton solution is given by \eqref{Bsoliton}
with $N=2$ and $M=4$, which yields 
\eqref{e:2DTL_ordinary2soliton},
with $\Theta_{ij}(n,y,t)=\theta_i(n,y,t)+\theta_j(n,y,t)$ and $S=\{(1,3),(1,4),(2,3),(2,4)\}$. 
Setting $\theta_{i,0}=0$ $\forall i=1,2,3,4$, the phase combinations are as follows:
\begin{gather*}
    \theta_1+\theta_3 = -X + 2y\cosh\gamma\cosh\alpha,\qquad
    \theta_2+\theta_4 = X + 2y\cosh\gamma\cosh\alpha,\\
     \theta_1+\theta_4 = 2y \cosh(\alpha+\gamma),\qquad
     \theta_2+\theta_3 = 2y\cosh(\alpha-\gamma),
\end{gather*}
where $X=2n\alpha-2\sinh\alpha\cosh\gamma t$. 
The coefficients appearing  in \eqref{e:2DTL_ordinary2soliton} can be expressed as
\begin{gather*}
        e^{P_4}-e^{P_2} = 2e^\alpha \sinh\gamma,\qquad
        e^{P_3}-e^{P_1} = 2e^{-\alpha} \sinh\gamma, \\
        e^{P_4}-e^{P_1} = 2 \sinh(\alpha+\gamma),\qquad 
        e^{P_3}-e^{P_2} = 2 \sinh(\gamma-\alpha),
\end{gather*}
so the tau function in \eqref{Bsoliton} reduces to:
\begin{equation*}
    \tau_n = 2e^{-\alpha}\sinh\gamma e^{-X+\mathcal Cy} + 2e^{\alpha}\sinh\gamma e^{X+\mathcal Cy}
        +2\sinh(\alpha+\gamma)\,e^{D_+y}+2\sinh(\gamma-\alpha)\,e^{D_-y}\,,
\end{equation*}
with $\mathcal C=2\cosh\gamma\cosh\alpha$, $D_+=2\cosh(\alpha+\gamma)$ and $D_-=2\cosh(\alpha-\gamma)$. 
Observing that $D_+=\mathcal C+2\sinh\gamma\sinh\alpha$ and $D_-=\mathcal C-2\sinh\gamma\sinh\alpha$, we can factor out the common term $e^{\mathcal{C}y}$ and write $\hat{\tau}_n:=e^{-\mathcal{C}y\tau_n}$. Since this factor is independent of the lattice index $n$, it cancels identically in the ratio $\tau_{n+1}\tau_{n-1}/\tau_n^2$ defining $B_n$  in \eqref{Bsoliton}. We may therefore equivalently work with
\begin{equation*}
    \hat\tau_n = 2e^{-\alpha}\sinh\gamma e^{-X}+2e^{\alpha}\sinh\gamma e^{X}+2\sinh{(\gamma-\alpha)}e^{-\mathcal Z}+2\sinh{(\alpha + \gamma)}e^{\mathcal Z},
\end{equation*}
with $\mathcal Z=2y\sinh\gamma\sinh\alpha$. 
After a few computations, one then obtains
\begin{equation*}
     \hat \tau_n = 4\sinh\gamma\cosh(X+\alpha)+4\sqrt{\sinh^2\gamma-\sinh^2\alpha}\,\cosh\left(\mathcal Z+\frac{1}{2}\log\bigg[\frac{\sinh(\gamma+\alpha)}{\sinh(\gamma-\alpha)}\bigg]\right).
\end{equation*}
Now we introduce for convenience the quantities 
\begin{equation*}
    \mathcal U:=X+\alpha,\qquad 
    \mathcal V:=\mathcal Z+\tfrac{1}{2}\log[\sinh(\gamma+\alpha)/\sinh(\gamma-\alpha)],\qquad 
    \rho:=\sqrt{\sinh^2{\gamma}-\sinh^2\alpha}\,.
\end{equation*}
Using the shift property $X(n\pm1)=X\pm 2 \alpha $ which implies $ \mathcal U(n\pm1)=\mathcal U\pm 2 \alpha $, the explicit two-soliton solution becomes~\eqref{e:2DTL_ordinary2soliton} with the constants given by the compact form~\eqref{SuperSol}.

\subsection{Asymmetric wedges, further differences between the 2DTL and the Kadomtsev-Petviashvili equation}
\label{app:E}

In the main text we considered symmetric wedges, i.e., ICs for which the slope for $y>0$ and that for $y<0$ were equal and opposite.  
In this appendix we consider asymmetric wedges, in which these two slopes are unrelated, and we show that, once the symmetries of the problem are taken into account, similar phenomena arise.
Along the way, we perform the same kind of analysis for the Kadomtsev-Petviashvili equation.

\paragraph{Kadomtsev-Petviashvili equation.}
Even though the KPII equation~\eqref{KP} is not invariant under rotations of the $xy$-plane, it possesses a ``pseudo-rotational'' symmetry (e.g., see \cite{RSPA2017v473p20160695})
\begin{equation}\label{PseudoRotation}
    \tilde U(x,y,t)\mapsto \tilde U(x+\zeta y-\zeta^2t,y-2\zeta  t,t)
\end{equation}
To see how this symmetry affects the line solitons, consider a generic exponential phase for the KP equation,
\begin{equation}
    \theta_j = \mathcal K_j x+ \mathcal K_j ^2y-\mathcal K_j^3t+\theta_{j,0}\,.
\end{equation}
Using the inverse transformation of \eqref{PseudoRotation}, namely,  
\begin{equation*}
    x=x'-\zeta y'+\zeta^2y, \qquad y=y'+2\zeta t,
\end{equation*}
one has
\begin{equation}
   \theta_j = \mathcal K_j x'+(\mathcal K_j^2-\zeta \mathcal K_j)y'+(\zeta^2\mathcal K_j+2\zeta \mathcal K_j^2-\mathcal K_j^3)t+\theta_{j,0}.
\end{equation}
Thus, the transformation leaves the coefficient of $x$ unchanged, while it changes the coefficient of $y$ from $\mathcal K^2$ to $\mathcal K^2-\zeta \mathcal K$.
Therefore, a generic slope parameter, which is the ratio between the $y$-coefficient and the $x$-coefficient, transforms as
\begin{equation}\label{invariance}
  \mathcal K_j \mapsto \mathcal K_j - \zeta\,.
\end{equation}
That is, the pseudo-rotation acts as a common translation of all slope parameters.

The above observation is corroborated by repeating the resonant interaction analysis of  Section \ref{s:wedge} for the KP equation~\eqref{KP}, and considering an asymmetric wedge initial condition
\begin{equation}\label{AsymmetricKP}
 \tilde U(x,y,0)= \begin{cases}
    \tilde U_-,  &x + q_1y < 0 ~\wedge~ y>0,\\
    \tilde U_-,  &x + q_2y < 0 ~\wedge~ y<0,\\
    \tilde U_+ \quad \text{otherwise},
    \end{cases}  
\end{equation}
with $q_1<0<q_2$. 
The symmetric wedge considered in the main text corresponds to the special case $q_2=-q_1$.

Let $a_0$ denote the common amplitude of the two incident line soliton. 
Following \cite{Ryskamp_2021} [specifically, equations~(13c) and~(14)], 
the corresponding soliton parameters for the dynamics of \eqref{AsymmetricKP} are
\begin{equation*}
   \mathcal K_1^-= \frac{q_1-\sqrt{a_0}}{2\sqrt{3}},\qquad 
   \mathcal K_1^+= \frac{q_1+\sqrt{a_0}}{2\sqrt{3}},\qquad  
   \mathcal K_2^-= \frac{q_2-\sqrt{a_0}}{2\sqrt{3}},\qquad  
   \mathcal K_2^+= \frac{q_2+\sqrt{a_0}}{2\sqrt{3}}.
\end{equation*}
In the subcritical regime, $q_2-q_1<2\sqrt{a_0}$, the ordering of the soliton parameters is $K_1^-<K_2^-<K_1^+<K_2^+$, 
and the resonant interaction generates a $[1,4]$ soliton stem, exactly as in the Toda analysis. 
Its amplitude is 
\begin{equation}
    A_{\mathrm{stem}}^{\mathrm{KP}}=\left (\sqrt{a_0}-\frac{(q_1-q_2)}{2}\right )^2,
\end{equation}
which depends only on the difference $q_1-q_2$, and is therefore the same as for a symmetric wedge.
(At the same time, note that, since the wedge angles are given by $q_j = \tan\beta_j$, the wedge opening angle is affected.) 
The corresponding stem slope is
\begin{equation}
q_{\mathrm{stem}}^{\mathrm{KP}}=\tfrac12(q_1+q_2),
\end{equation}
which is the mean orientation of the wedge. 
This dependence is entirely associated with the KP pseudo-rotational invariance \eqref{PseudoRotation} and can be removed by \eqref{invariance} choosing $\zeta=-\frac12(q_1+q_2)$. 
Therefore, asymmetric wedges sharing the same opening angle are dynamically equivalent up to a symmetry transformation of the KP equation.
This discussion shows that every asymmetric KPII wedge can be reduced to a symmetric one by means of the pseudo-rotational symmetry.

\paragraph{Two-dimensional Toda lattice.}
It is natural to ask whether a reduction similar to the one above also exists for the two-dimensional Toda lattice.
The answer is affirmative, although the underlying symmetry is of a different nature. Indeed, the 2DTL Eq. \eqref{eq_Toda1} is invariant under Lorentz boosts in the $yt$-plane:
\begin{equation}\label{Lorentz}
y'=y\cosh\zeta+t\sinh\zeta,
\qquad
t'=t\cosh\zeta+y\sinh\zeta,
\end{equation}
For a 2DTL line soliton, the phase variables are written naturally as  
$\theta_j(n,y,t) = y\cosh{P_j}+nP_j+\sigma t\sinh{P_j}+\theta_{j,0}$,
as per Eq.~\eqref{phase}.
Under the Lorentz transformation~\eqref{Lorentz},
one finds $P_i\mapsto P_i-\sigma \zeta.$
Thus the Lorentz symmetry acts as a common translation of all phase parameters $P_i$. 
Recalling that the 2DTL line solitons are parameterized by parameters
$P_1=\gamma-\alpha$ and $P_2=\gamma+\alpha$, the Lorentz boost preserves the parameter $\alpha = \frac12(P_2-P_1)$ while it shifts the parameter $\gamma=\frac12(P_1+P_2)$ according to $\gamma \mapsto \gamma-\sigma \zeta $. Consequently, an asymmetric wedge with parameters $\gamma_1$ and $\gamma_2$, where $\gamma_1$ and $\gamma_2$ denote the values of the parameter $\gamma$ associated with the two incident line solitons, can be transformed into a symmetric configuration by choosing 
\begin{equation}\label{zetaToda}
\zeta=\tfrac12\sigma(\gamma_1+\gamma_2),
\end{equation}
for which 
\begin{equation}\label{FinalGamma}
      \gamma_1'=\tfrac12(\gamma_1-\gamma_2),\qquad 
    \gamma_2'=-\tfrac12(\gamma_1-\gamma_2).
\end{equation}
Thus, exactly as in the KP case, any asymmetric wedge can be mapped to a symmetric one. 
However, the action of the symmetry is fundamentally different from that in the KP equation. 
In the latter, the Galilean symmetry acts directly on the physical slope parameter,
$q\mapsto q+\zeta$,
so that the opening angle completely characterizes the wedge modulo symmetry. In contrast, the Toda Lorentz symmetry acts linearly on the parameter $\gamma$, while the physical slope is related to $\gamma$ through the nonlinear relation
\begin{equation}\label{NonlinearRelation}
q=\frac{\sqrt{A_V}}{\alpha}\sinh{\gamma}\end{equation} 
Therefore, the symmetry does not correspond to a uniform shift of the physical slopes.

Motivated by the previous discussion, we now consider asymmetric wedges in the 2DTL. Although the Lorentz symmetry can always be used to symmetrize the parameters $\gamma_1$ and $\gamma_2$, it is still useful to formulate the problem directly in terms of the physical slopes $q_1$ and $q_2$ appearing in the initial data.  That is, we consider ICs of the form
\begin{equation}\label{Asymmetric}
 B_n(y,0)= \begin{cases}
    B_-,  &n + q_1y < 0 ~\wedge~ y>0,\\
    B_-,  &n + q_2y < 0 ~\wedge~ y<0,\\
    B_+ \quad \text{otherwise},
    \end{cases}  
\end{equation}
where $q_1q_2<0$ and\ $q_1\neq -q_2$. 
As in the main text, for simplicity and without loss of generality, we will consider the case $B_+=0$. 
The two oblique DSWs still have the same amplitude $A_{B}^{\mathrm{obl}}$
(since, as observed in \ref{obs: AmplitudeOblique}, the latter depends only on the jump size~$\Delta$), but different slopes $q_1$ and $q_2$. 
The corresponding soliton parameters, obtained from the inversion formulae~\eqref{InverseRelation}, are:
\begin{equation}
    (P_1,P_2,P_3,P_4):=\mathrm{ord}\left \{\gamma_1-\alpha, \gamma_1+\alpha, \gamma_2-\alpha,\gamma_2+\alpha \right \}
\end{equation}
with $\alpha=\arcsinh{\left( \sqrt{A_V}\right)}$ as before and 
\begin{equation*}
    \gamma_1=\arcsinh\left(q_1\frac{\arcsinh (\sqrt{A_V})}{\sqrt{A_V}} \right), \qquad
      \gamma_2=\arcsinh\left(q_2\frac{\arcsinh (\sqrt{A_V})}{\sqrt{A_V}} \right).
\end{equation*}
A similar analysis to the one in Section~\ref{s:wedge} then yields:

\begin{remark}
\label{Obs:AsymmetricCases} 
The asymmetric wedge \eqref{Asymmetric} gives rise to precisely  four dynamical cases analogous to the ones described in Section \ref{s:wedge}:
\begin{description}
\item[\textbf{Case 1:}] 
$q_1<0<q_2 $ and $\gamma_2-\gamma_1<2\alpha$. In this case $(P_1,P_2,P_3,P_4)=(\gamma_1-\alpha, \  \gamma_2-\alpha, \ \gamma_1+\alpha, \ \gamma_2+\alpha)$ that corresponds to the subcritical type II wedge. 
\item[\textbf{Case 2:}] 
$q_1<0<q_2 $  and $\gamma_2-\gamma_1>2\alpha$. In this case $(P_1,P_2,P_3,P_4)=(\gamma_1-\alpha, \  \gamma_1+\alpha, \ \gamma_2-\alpha, \ \gamma_2+\alpha)$ that corresponds to the supercritical type II wedge. 
\item[\textbf{Case 3:}]
$q_2<0<q_1 $  and $\gamma_1-\gamma_2<2\alpha$. In this case $(P_1,P_2,P_3,P_4)=(\gamma_2-\alpha, \  \gamma_1-\alpha, \ \gamma_2+\alpha, \ \gamma_1+\alpha)$ that corresponds to the subcritical type I wedge.
\item[\textbf{Case 4:}]
$q_2<0<q_1 $  and $\gamma_1-\gamma_2>2\alpha$. In this case $(P_1,P_2,P_3,P_4)=(\gamma_2-\alpha, \  \gamma_2+\alpha, \ \gamma_1-\alpha, \ \gamma_1+\alpha)$  that corresponds to the supercritical type I wedge.
\end{description}   

In all four cases, the resonance condition $P_2=P_3$ is equivalent to $|\gamma_2-\gamma_1|=2\alpha$, 
which yields
\begin{equation} \label{AsymmetricResonantCodnition}
  |\gamma_2-\gamma_1|=2\alpha\,,
\end{equation}
which in turn generalizes the symmetric resonance condition~\eqref{Eq:Res_qA}. 
Even in the asymmetric wedge setting, the interaction of the leading edge solitons is either ordinary or resonant, and the asymmetric soliton interactions of Section~\ref{s:2Sol} does not arise in any of the four cases described.
The asymmetry of the wedge only affects the slope and amplitude of the stem, but not its qualitative nature. 
This is a direct consequence of the parameter ordering structure, which always produces index pairs of the form $[1,3]$ and $[2,4]$ or $[1,2]$ and $[3,4]$, regardless of the asymmetry between $\gamma_1$ and $\gamma_2$. These predictions are illustrated in Figs.~\ref{fig:AsII} and~\ref{fig:AsI}, which show respectively the time evolution of asymmetric type~II and type~I wedges.
\end{remark}

\begin{figure}[t!]
    \centering
    \includegraphics[width=\linewidth]{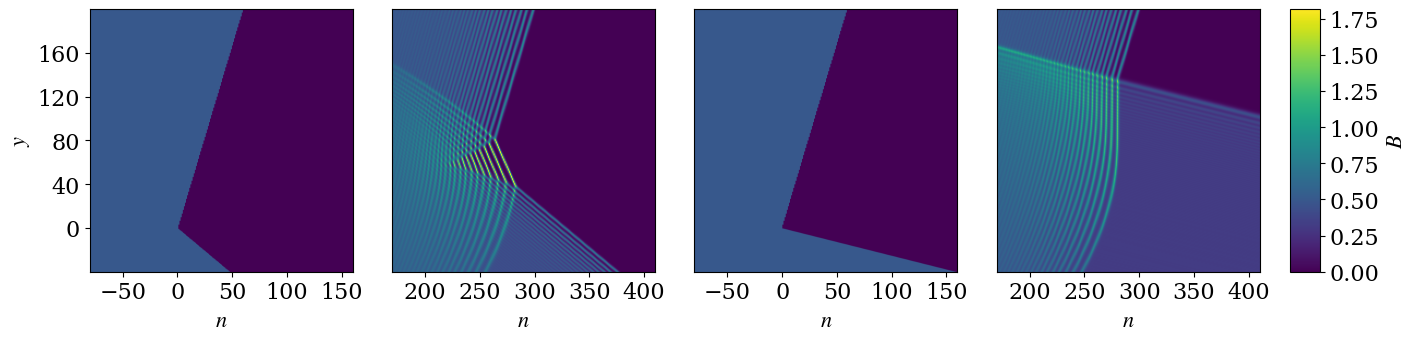}
    \caption{Density plots illustrating the evolution of an asymmetric type II wedges with $q_1=-0.3$ and $q_2=-1.2,-4$, at $t=0$ and $t=300$. Left two panels: subcritical regime ($q_1=-0.3$,$q_2=-1.2$), showing the formation of a tilted stem DSW  with slope $\tilde q_{\mathrm{stem}}\neq 0$, in contrast to the symmetric case. Right two panels: supercritical regime (($q_1=-0.3$,$q_2=-4$)), showing the formation of a localized peak propagating along a tilted trajectory.}
    \label{fig:AsII}
\bigskip
    \centering
    \includegraphics[width=\linewidth]{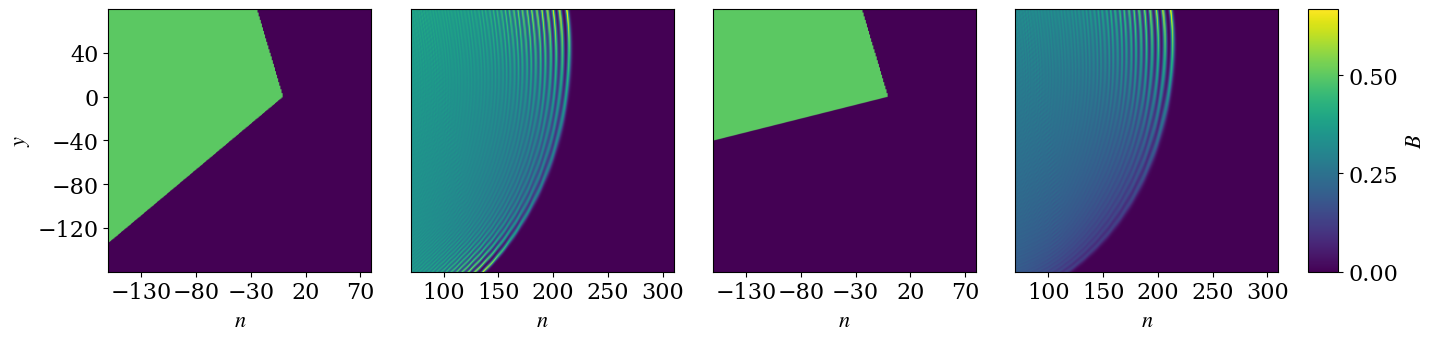}
    \caption{Density plots illustrating the evolution of an asymmetric type I wedges with $q_1=0.3$ and $q_2=1.2,4$, at $t=0$ and $t=300$. Left two panels: subcritical regime ($q_1=0.3$,$q_2=1.2$). Right two panels: supercritical regime ($q_1=0.3$,$q_2=4$).}
    \label{fig:AsI}
\end{figure}

Since the interactions are still described by exact multi-soliton solutions of the two-dimensional Toda lattice, all quantities can be obtained by following the same arguments used for the symmetric case. The only difference is that the two incident oblique DSWs are now characterized  by distinct parameters $\gamma_1$ and $\gamma_2$. 
Therefore, for brevity, we will not repeat the full derivation of Section~\ref{s:wedge} here. 
On the other hand, 
as an illustration, we consider a subcritical type~II wedge, for which the parameters are $(P_1,P_2,P_3,P_4)=(\gamma_1-\alpha, \  \gamma_2-\alpha, \ \gamma_1+\alpha, \ \gamma_2+\alpha)$.
The resonant interaction generates a stem corresponding to the $[1,4]$ soliton with slope and amplitude given by
\begin{equation} \label{parameters_stem_as} q_{\mathrm{stem}}^{\mathrm{Toda}}=\frac{\cosh{(\gamma_2+\alpha)}-\cosh{(\gamma_1-\alpha)}}{\gamma_2-\gamma_1+2\alpha},\qquad
A_{B,\mathrm{stem}}^{\mathrm{Toda}}=\log\left[1+\sinh^2\big(\tfrac12(\gamma_2-\gamma_1+2\big)\right].
\end{equation}
The asymmetric nature of the wedge is reflected in the slope of the stem. Indeed, the expression for the stem slope in \eqref{parameters_stem_as} can be rewritten as
\begin{equation*}
q_{\mathrm{stem}}^{\mathrm{Toda}}=\sinh\big(\tfrac12(\gamma_2+\gamma_1)\big)
    \frac{\sinh\left(\frac12(\gamma_2-\gamma_1)+\alpha\right)}
        {\tfrac12(\gamma_2-\gamma_1)+\alpha}\,.
\end{equation*}
Thus, the stem slope does not depend only on the wedge opening parameter $\gamma_1-\gamma_2$, but also on the mean orientation $\gamma_1+\gamma_2$. 
In the symmetric case $\gamma_2=-\gamma_1$, one has $q_{\mathrm{stem}}^{\mathrm{Toda}}=0$, so that the stem is vertical, 
whereas for a genuinely asymmetric wedge, the stem is tilted. 
On the other hand, this tilt is not a new feature of the dynamics,
since, as discussed above, the Lorentz symmetry acts as a common translation of the soliton parameters. 
Indeed, by choosing $\zeta$ as in~\eqref{zetaToda}, the pair $ (\gamma_1,\gamma_2)$ is mapped into $(\gamma'_1,\gamma'_2)$ via~\eqref{FinalGamma}, 
thereby allowing one to reduce an arbitrary asymmetric wedge to a symmetric one. 
In the transformed frame, the stem satisfies $\gamma'_\mathrm{stem}=0$, while its amplitude, which depends only on the invariant difference $\gamma_2-\gamma_1$, remains unchanged.

To summarize, both the KP equation and the 2DTL possess a symmetry that allows one to reduce asymmetric wedges into symmetric ones. 
The difference lies in the variables on which these symmetries act. 
For the KP equation, the pseudo-rotational symmetry acts directly on the physical slope parameters $q_i$. In contrast, the Lorentz symmetry of the 2DTL acts on the parameters $\gamma_i$, while the physical slopes are recovered through the nonlinear relation \eqref{NonlinearRelation}.

\addcontentsline{toc}{section}{References}
\printbibliography

\end{document}